\documentclass[11pt]{amsart}
\usepackage[a4paper,margin=2cm]{geometry}

\usepackage{graphicx}
\usepackage[T1]{fontenc}
\usepackage{lmodern}
\usepackage[breaklinks]{hyperref}

\usepackage[ruled]{algorithm}
\usepackage[noend]{algpseudocode}
\usepackage{enumerate,hyperref,verbatim}

\DeclareMathOperator{\PC}{h}
\DeclareMathOperator{\BC}{h}

\newcommand{\CPref}[1]{(\hyperref[CP#1]{CP#1})}
\newcommand{\Mref}[1]{(\hyperref[M#1]{M#1})}

\newcommand{\bound}{\ensuremath{N}}

\newcommand{\presentletters}{\textnormal{\textsl{Letters}}}
\newcommand{\pairs}{\textnormal{\textsl{Pairs}}}
\newcommand{\crossing}{\textnormal{\textsl{Crossing}}}
\newcommand{\ncrossing}{\textnormal{\textsl{Non-Crossing}}}

\newcommand{\algradix}{\algofont{RadixSort}}
\newcommand{\algpop}{\algofont{Pop}}
\newcommand{\algtestsimple}{\algofont{TestSimpleSolution}}
\newcommand{\algtest}{\algofont{TestSolution}}
\newcommand{\algprefsuff}{\algofont{CutPrefSuff}}
\newcommand{\algpreproc}{\algofont{PreProc}}

\newcommand{\first}[1][{[X]}]{\ensuremath{\textnormal{\textsl{first}}#1}}
\newcommand{\last}[1][{[X]}]{\ensuremath{\textnormal{\textsl{last}}#1}}

\newcommand{\algpair}{\algofont{PairCompNCr}}
\newcommand{\algpairc}{\algofont{PairComp}}

\newcommand{\algsonevar}{\algofont{OneVarWordEq}}
\newcommand{\algsolveeq}{\algofont{WordEqSat}}

\newcommand{\algblocks}{\algofont{BlockCompNCr}}
\newcommand{\algblocksc}{\algofont{BlockComp}}

\newcommand{\algofont}[1]{\textnormal{\textsc \selectfont\sffamily #1}}

\newcommand{\sol}[1]{\ensuremath{S(#1)}}
\newcommand{\solution}{\ensuremath{S}}

\newcommand{\letters}{\ensuremath{\Gamma}}

\newcommand{\NPclass}{{\sf NP}}
\newcommand{\Pclass}{{\sf P}}
\newcommand{\PSPACE}{{\sf PSPACE}}

\newcommand{\NEXPTIMEclass}{{\sf NEXPTIME}}

\newcommand{\poly}{{\sf {poly}}}

\newtheorem{theorem}{Theorem}
\newtheorem{lemma}{Lemma}

\theoremstyle{remark}

\theoremstyle{definition}

\providecommand{\Ocomp}{\mathcal{O}}

\newcommand{\twodots}{\mathinner{\ldotp\ldotp}}

\title{One-variable word equations in linear time}

\author[A. Je\.z]{Artur Je\.z}
\address{
Max Planck Institute f\"ur informatik,
Saarbr\"ucken, Germany\\
on leave from:
Institute of Computer Science, 
University of Wroc\l{}aw,
Wroc\l{}aw, Poland
}

\thanks{This work was supported by Alexander von Humboldt Foundation.}



\begin{document}

\begin{abstract}
In this paper we consider word equations with one variable (and arbitrary many occurrences of it).
A recent technique of recompression, which is applicable to general word equations,
is shown to be suitable also in this case.
While in general case the recompression is non-deterministic it determinises in case of one variable and the obtained running time
is $\Ocomp(n + \#_X \log n)$, where $\#_X$ is the number of occurrences of the variable in the equation.
This matches the previously-best algorithm due to D\k{a}browski and Plandowski.
Then, using a couple of heuristics as well as more detailed time analysis, the running time is lowered to $\Ocomp(n)$ in the RAM model.
Unfortunately, no new properties of solutions are shown.
\end{abstract}

\maketitle

\section{Introduction}

\subsection{Word equations}
The problem of satisfiability of word equations was considered as one of the most intriguing in computer science
and its study was initiated by Markow already in the '50.
The first algorithm for it was given by Makanin~\cite{makanin}, despite earlier conjectures that the problem is undecidable.
The proposed solution was very complicated in terms of proof-length, algorithm and computational complexity.
It was improved several times, however, no essentially different approach was proposed for over two decades.

An alternative algorithm was proposed by Plandowski and Rytter~\cite{PlandowskiICALP},
who shoved that each minimal solution of a word equation is exponentially compressible,
in the sense that for a word equation of size $n$ and minimal solution of size $N$
the LZ77 (a popular practical standard of compression) representation of the minimal solution
is polynomial in $n$ and $\log N$.
Hence a simple non-deterministic algorithm that guesses a compressed representation of a solution
and verifies the guess has running time polynomial in $n$ and $\log N$.
However, at that time the only bound on $N$ followed from Makanin's work (with further improvements)
and it was triply exponential in $n$.

Soon after Plandowski showed, using novel factorisations, that $N$ is at most doubly exponential~\cite{PlandowskiSTOC},
showing that satisfiability of word equations is in \NEXPTIMEclass.
Exploiting the interplay between factorisations and compression he improved the algorithm so that it worked in \PSPACE~\cite{PlandowskiFOCS}.

Producing a description of all solutions of a word equation, even when a procedure for verification of its satisfiability is known,
proved to be also a non-trivial task. Still, it is also possible to do this in \PSPACE~\cite{PlandowskiSTOC2},
though insight and non-trivial modifications to the earlier procedure are needed.

On the other hand, it is only known that the satisfiability of word equations is \NPclass-hard.

\subsubsection{Two variables}
Since in general the problem is outside \Pclass, it was investigated, whether some subclass is feasible,
with a restriction on the number of variables being a natural candidate.
It was shown by Charatonik and Pacholski~\cite{CharatonikPacholski} that indeed, when only two variables are allowed
(though with arbitrarily many occurrences), the satisfiability can be verified in deterministic polynomial time.
The degree of the polynomial was very high, though.
This was improved over the years and the best known algorithm is by D\k{a}browski and Plandowski~\cite{twovarnew}
and it runs in $\Ocomp(n^5)$ and returns a description of all solutions.

\subsubsection{One variable}
Clearly, the case of equations with only one variable is in \Pclass.
Constructing a cubic algorithm is almost trivial, small improvements are needed to guarantee a quadratic running time.
First non-trivial bound was given by Obono, Goralcik and Maksimenko,
who devised an $\Ocomp(n \log n)$ algorithm~\cite{onevarfirst}.
This was improved by D\k{a}browksi and Plandowski~\cite{onevarold} to $\Ocomp(n + \#_X \log n)$,
where $\#_X$ is the number of occurrences of the variable in the equation.
Furthermore they showed that there are at most $\Ocomp(\log n)$ distinct solutions and at most one infinite family of solutions.
Intuitively, the $\Ocomp(\#_X \log n)$ summand in the running time comes from the time needed to find and test these $\Ocomp(\log n)$ solutions.

This work was not completely model-independent, as it assumed that the alphabet $\letters$ is finite or that
it can be identified with numbers.
A more general solution was presented by Laine and Plandowski~\cite{onevarnew},
who improved the bound on the number of solutions to $\Ocomp(\log \#_X)$ 
(plus the infinite family) and gave an $\Ocomp(n \log \#_X)$ algorithm that runs in a pointer machine model
(i.e.\ letters can be only compared and no arithmetical operations on them are allowed);
roughly one candidate for the solution is found and tested in linear time.

\subsection{Recompression}
Recently, the author proposed a new technique of \emph{recompression} based on previous techniques of Mehlhorn et.~al\cite{MehlhornSU97}
(for dynamic text equality testing), Lohrey and Mathissen~\cite{LohreySLP} (for fully compressed membership problem for NFAs)
and Sakamoto~\cite{SLPaproxSakamoto} (for construction of the smallest grammar for the input text).
This method was successfully applied to various problems related to grammar-compressed strings~\cite{fullyNFA,FCPM,grammar}.
Unexpectedly, this approach was also applicable to word equations,
in which case alternative proofs of many known algorithmic results were obtained using a unified approach~\cite{wordequations}.

The technique is based on iterative application of two replacement schemes performed on the text $t$:
\begin{description}
	\item[pair compression of $ab$]
	For two different letters $a$, $b$ such that substring $ab$ occurs in $t$
	replace each of $ab$ in $t$ by a fresh letter $c$.
	\item[$a$'s block compression]
	For each maximal block $a^\ell$, where $a$ is a letter and $\ell >1$,
	that occurs in $t$, replace all $a^\ell$s in $t$ by a fresh letter $a_\ell$.
\end{description}

In one phase, pair compression (block compression) is applied to all pairs (blocks, respectively)
that occurred at the beginning of this phase.
Ideally, each letter is then compressed and so the length of $t$ halves,
in a worst-case scenario during one phase $t$ is still shortened by a constant factor.

The surprising property is that such a schema can be efficiently applied to grammar-compressed data~\cite{fullyNFA,FCPM}
or to text given in an implicit way, i.e.\ as a solution of a word equation~\cite{wordequations}.
In order to do so, local changes of the variables (or nonterminals) are needed: $X$ is replaced with $a^\ell X$ (or $Xa^\ell$),
where $a^\ell$ is prefix (suffix, respectively) of the substitution for $X$.
In this way the solution that substitutes $a^\ell w$ (or $wa^\ell$, respectively) for $X$ is implicitly replaced with one that substitutes $w$.

\subsubsection{Recompression and one-variable equations}
Clearly, as the recompression approach works for general word equations, it can be applied also to restricted subclasses.
However, while in case of word equations it heavily relies on the nondeterminism,
when restricted to instances with one variable it can be easily determinised;
Section~\ref{sec:prelim} recalls the main notions of word equations and recompression.
Furthermore, a fairly natural implementation has $\Ocomp(n + \#_X \log n)$ running time,
so the same as the D\k{a}browski and Plandowski algorithm~\cite{onevarold};
this is presented in Section~\ref{sec:main algorithm}.
Furthermore adding a few heuristics, data structures as well as applying a more sophisticated analysis yields a linear running time,
this is described in Section~\ref{sec:faster}.

\subsection{Outline of the algorithm}
In this paper we present an algorithm for one-variable equation based on the recompression.
It also provides a compact description of all solutions of such an equation.
Intuitively: when pair compression is applied, say $ab$ is replaced by $c$
(assuming it \emph{can} be applied)
then there is a one-to-one correspondence of the solutions before and after the compression,
this correspondence is simply an exchange of all $ab$s by $c$s and vice-versa.
The same applies to the block compression.
On the other hand, the modification of $X$ can lead to loss of solutions
(note that for technical reasons we do note consider the solution $\sol X = \epsilon$):
when $X$ is to be replaced with $a^\ell X$ the new equation has corresponding solutions
for \solution{} \emph{other than} $\sol X = a^\ell$.
So before the replacement, it is tested whether $\sol X = a^\ell$ is a solution
and if so, it is reported.
The test itself is simple: both sides of the equation are read and their values under substitution
$\sol X = a^\ell$ are created on the fly and compared symbol by symbol,
until a mismatch is found or both strings end.

It is easy to implement the recompression so that one phase takes linear time.
Then the cost can be distributed to explicit words between the variables, each of them is charged proportionally to its length.
Consider such a string $w$, 
if it is long enough, its length decreases by a constant factor in one phase, see Lemma~\ref{lem:reducing length}.
Thus, the cost of compressing this fragment and testing a solution can be charged to the lost length.
However, this is not true when $w$ is short and the $\#_X \log n$ summand in the running time comes from
bounding the running time for such `short' strings.

In Section~\ref{sec:faster} it is shown that using a couple of heuristics as well as more involved analysis
the running time can be lowered to $\Ocomp(n)$.
The mentioned heuristics are as follows:
\begin{itemize}
	\item The problematic `short' words between the variables need to be substrings of the `long' words,
	this allows smaller storage size and consequently faster compression.
	\item When we compare $X w_1 X w_2 \ldots w_m X$ from one side of the equation with its copy
	(i.e.\ another occurrence $X w_1 X w_2 \ldots w_m X$) on the other side,
	we make such a comparison in $\Ocomp(1)$ time (using suffix arrays).
	\item $(\sol{X}u)^m$ and $(\sol{X}u')^{m'}$ (perhaps offsetted) are compared 
	in $\Ocomp(|u| + |u'|)$ time instead of naive $\Ocomp(m \cdot |u| + m' \cdot |u'|)$,
	using simple facts from combinatorics on words.
\end{itemize}
Furthermore a more insightful analysis shows that problematic `short' words in the equation
can be used to invalidate several candidate solutions fast,
even before a mismatch in the equation is found during the testing.
This allows a tighter estimation of the time spent on testing the solutions.

\subsubsection*{A note on the computational model}
In order to perform the recompression efficiently, some algorithm for grouping pairs is needed.
When we can identify the symbols in $\letters$ with consecutive numbers,
the grouping can be done using \algradix{} in linear time.
Thus, all (efficient) applications of recompression technique make such an assumption.
On the other hand, the second of the mentioned heuristics craves checking string equality in constant time,
to this end a suffix array~\cite{suffixarrays} plus a structure for answering \emph{longest common prefix query}
(lcp) is employed~\cite{lcpsuffixarrays} on which we use range minimum queries~\cite{rmq}.
The last structure needs the flexibility of the RAM model to run in $\Ocomp(1)$ time per query.

\section{Preliminaries}
\label{sec:prelim}

\subsection{One-variable equations}
\label{sec:onevar}
A \emph{word equation} with one variable over the alphabet $\letters$ and variable $X$
is `$\mathcal A = \mathcal B$', where $\mathcal A, \mathcal B \in (\letters \cup \{ X\})^*$.
During the run of algorithm \algsolveeq{} we introduce new letters into \letters,
but no new variable is introduced.
In this paper we shall consider only equations with one variable.

Without loss of generality in a word equation $\mathcal A = \mathcal B$ one of $\mathcal A$ and $\mathcal B$ begin
with a variable and the other with a letter:
\begin{itemize}
	\item if they both begins with the same symbol (be it letter or nonterminal),
	we can remove this symbol from them, without affecting the set of solutions;
	\item if they begin with different letters, this equation clearly has
	no solution.
\end{itemize}
The same applies to the last symbols of $U$ and $V$.
Thus, in the following we assume that the equation is of the form
\begin{equation}
\label{eq:univariate}
A_0 X A_1 \dots A_{n_{\mathcal A}-1} X A_{n_{\mathcal A}} = 
X B_1 \dots B_{n_{\mathcal B}-1} X B_{n_{\mathcal B}} \enspace ,
\end{equation}
where $A_i, B_i \in \letters^*$ (we call them \emph{words} or \emph{explicit words})
and $n_{\mathcal A}$ ($n_{\mathcal B}$) denote the number of $X$ occurrences in $\mathcal A$
($\mathcal B$, respectively).
Note that exactly one of $A_{n_{\mathcal A}}$, $B_{n_{\mathcal B}}$ is empty and $A_0$ is non-empty.
If this condition is violated for any reason, we greedily repair it by cutting identical letters (or variables) from both sides of the equation.
We say that $A_0$ is the \emph{first} word of the equation and the non-empty of $A_{n_{\mathcal A}}$ and $B_{n_{\mathcal B}}$
is the \emph{last word}.
We additionally assume that none of words $A_i, B_j$ is empty. We later (after Lemma~\ref{lem:pop preserves solutions})
justify why this is indeed without loss of generality.

A \emph{substitution} \solution{} assigns a string to $X$, we expand it to $(X \cup \letters)^*$
with an obvious meaning. A \emph{solution} is a substitution such that $\sol {\mathcal A} = \sol{\mathcal B}$.
For a given equation $\mathcal A = \mathcal B$ we are looking for a description of all its solutions.
We treat the empty solution $\sol X = \epsilon$ in a special way and always assume that $\sol X \neq \epsilon$.

Note that if $\sol X \neq \epsilon$, then using~\eqref{eq:univariate}
we can always determine the first ($a$) and last ($b$) letter of $\sol X$ in $\Ocomp(1)$ time.
In fact, we can determine the length of the $a$-prefix and $b$-suffix of \sol X.

\begin{lemma}
\label{lem:a prefix}
For every solution \solution{} of a word equation the first letter of \sol X is the first letter of $A_0$
and the last the last letter of $A_{n_{ \mathcal A}}$ or $B_{n_{\mathcal B}}$ (whichever is non-empty).

If $A_0 \in a^+$ then $\sol X \in a^+$ for each solution \solution{} of $\mathcal A = \mathcal B$.

If the first letter of $A_0$ is $a$ and $A_0 \notin a^+$
then there is at most one solution $\sol X \in a^+$,
existence of such a solution can be tested (and its length returned) in $\Ocomp(|\mathcal A| + |\mathcal B|)$ time.
Furthermore, for $\sol X \notin a^+$ the lengths of the $a$-prefixes of \sol X and $A_0$ are the same.
\end{lemma}
It is later shown that finding all solutions from $a^+$ can be done in linear time, see Lemma~\ref{lem:a* solution}.
\begin{proof}
Concerning the first claim, observe that the first letter of $\sol {\mathcal A}$ is the first letter of $A_0$,
while the first letter of $\sol {\mathcal B}$ is the first letter of \sol X, hence those letters are equal.
The same applies to the last letter of \sol X and the last letter of $A_{n_{\mathcal A}}$ or $B_{n_{\mathcal B}}$,
whichever of them is non-empty.

Consider the case when $A_0 \in a^+$ and suppose that $\sol X \notin a^*$, let $\ell \geq 0$ be the length
of the $a$-prefix of \sol X.
The length of the $a$-prefix of \sol {\mathcal A} is then $|A_0| + \ell > \ell$,
which is the length of the $a$-prefix of \sol {\mathcal B}, contradiction.
Hence $\sol X \in a^+$.

Consider now the case when $A_0$ begins with $a$ but $A_0 \notin a^+$, let its $a$-prefix has length $\ell_A$.
Consider $\sol X \in a^+$, say $\sol X = a^\ell$.
Let the first letter other than $a$ in $\mathcal B$ be the $\ell_{B}+1$ letter in $\mathcal B$
and let it be in explicit word $B_i$.
If there is no such $B_i$ then there is no solution $\sol X \in a^+$, as then \sol{\mathcal B}
consists only of $a$s, which is not true for \sol{\mathcal A}.
The length of the $a$-prefix of \sol {\mathcal A} is $\ell_A$,
while the length of the $a$-prefix of \sol {\mathcal B} is $\ell_B + i \cdot \ell$.
Those two need to be equal, so $\ell_A = \ell_B + i \cdot \ell $ and consequently
$\ell = \frac{\ell_A - \ell_B}{i}$, so this is the only candidate for the solution.

It is easy to verify whether $\sol X = a^\ell$ is a solution for a single $\ell$ in linear time.
It is enough to compare $\sol {\mathcal A}$ and $\sol {\mathcal B}$ letter by letter,
note that they can be created on the fly while reading ${\mathcal A}$ and ${\mathcal B}$.
Each such comparison consumes one symbol from ${\mathcal A}$ and ${\mathcal B}$
(note that if we compare a suffix of \sol X, i.e.\ some $a^{\ell'}$ for $\ell'<\ell$,
with $\sol X = a^\ell$ we simply remove $a^{\ell'}$ from both those strings).
So the running time is linear.

Lastly, consider $\sol X \notin a^*$.
Then the $a$-prefix of \sol{\mathcal A} has length $\ell_A$ and as $\sol X \notin a^+$,
the $a$-prefix of \sol {\mathcal B} is the same as the $a$-prefix of \sol X,
which consequently has length $\ell_A$.
\end{proof}
Symmetric version of Lemma~\ref{lem:a prefix} holds for the suffix of \sol X.

By $\algtestsimple(a)$ we denote a procedure, described in Lemma~\ref{lem:a prefix},
that for $A_0 \notin a^*$ establishes the unique possible solution $\sol X = a^\ell$,
tests it and returns $\ell$ if this indeed is a solution.

\subsection{Representation of solutions}
Consider any solution \solution{} of $\mathcal A = \mathcal B$.
We claim that \sol X is uniquely determined by its length and so when describing solution of $\mathcal A = \mathcal B$
it is enough to give their lengths.

\begin{lemma}
	\label{lem: solution form}
Each solution \solution{} of equation of the form~\eqref{eq:univariate} is of the form $\sol X = (A_0)^k A$,
where $A$ is a prefix of $A_0$ and $k \geq 0$. In particular, it is uniquely defined by its length. 
\end{lemma}
\begin{proof}
If $|\sol X| \leq |A_0|$ then $\sol X$ is a prefix of $A_0$, so \sol X is uniquely determined by its length.
When $|\sol X| > |A_0|$ then \sol {\mathcal A} begins with $A_0 \sol X$ while \sol {\mathcal B}
begins with $\sol X$ and thus \sol X has a period $A_0$. Consequently, it is of the form $A_0^kA$, where $A$ is a prefix of $A_0$.
\qedhere
\end{proof}

\subsubsection*{Weight}
Each letter in the current instance of our algorithm represents some string (in a compressed form) of the input equation,
we store its \emph{weight} which is the length of such a string.
Furthermore, when we replace $X$ with $a^\ell X$ (or $Xa^\ell$) we keep track of sum of weights of all letters removed so far
from $X$.
In this way, for each solution of the current equation we know what is the length of the corresponding solution of the original equation (it is the sum of weights of letters removed so far from $X$ and the weight of the current solution).
Therefore, in the following, we will not explain how we recreate the solutions of the original equation from the solution
of the current one.

\subsection{Recompression}
We recall here the technique of recompression~\cite{fullyNFA,FCPM,wordequations}, restating all important facts about it.
Note that in case of one variable many notions simplify.

\subsubsection{Preserving solutions}
All subprocedures of the presented algorithm should preserve solutions,
i.e.\ there should be a one-to-one correspondence between solution before and after
the application of the subprocedure.
However, when we replace $X$ with $a^\ell X$ (or $Xb^r$),
some solutions may be lost in the process and so they should be reported.
We formalise these notions.

We say that a subprocedure \emph{preserves solutions}
when given an equation $\mathcal A = \mathcal B$ it returns
$\mathcal A' = \mathcal B'$ such that for some strings $u$ and $v$ (calculated by the subprocedure)
\begin{itemize}
	\item some solutions of $\mathcal A = \mathcal B$ are reported by the subprocedure;
	\item for each unreported solution \solution{} of $\mathcal A = \mathcal B$ there is a solution
	$\solution'$ of $\mathcal A' = \mathcal B'$, where $\sol X = u \solution'(X) v$ and
	$\sol{\mathcal A} = u\solution' (\mathcal A')v$;
	\item for each solution $\solution'$ of $\mathcal A' = \mathcal B'$ the $\sol X = u \solution'(X) v$
	is an unreported solution of $\mathcal A = \mathcal B$.
\end{itemize}
The intuitive meaning of these conditions is that during transformation of the equation, either we report a solution
or the new equation has a corresponding solution (and no new `extra' solutions).

By $\PC_{c \to ab }(w)$ we denote the string obtained from $w$ by replacing each $c$ by $ab$,
which corresponds to the inverse of pair compression.
We say that a subprocedure \emph{implements pair compression} for $ab$,
if it satisfies the conditions for preserving solutions above,
but with $\sol X = u \PC_{c \to ab}(\solution'(X)) v$ and $\sol {\mathcal A} = u \PC_{c \to ab}(\solution'(\mathcal A')) v$
replacing $\sol X = u \solution'(X) v$ and $\sol{\mathcal A} = u \solution'(\mathcal A') v$.
Similarly, by $\BC_{a^\ell \to a^\ell}(w)$ we denote a string with letters $a^\ell$ replaced with blocks $a^\ell$
(note that this requires that we know, which letters `are' $a_\ell$ and what is the value of $\ell$, but this is always clear from the context)
and we say that a subprocedure \emph{implements blocks compression} for a letter $a$.
The intuitive meaning is the same as in case of preserving solutions: we not loose, nor gain any solutions.

Given an equation $\mathcal A = \mathcal B$, its solution \solution{}
and a pair $ab \in \letters^2$ occurring \sol U (or \sol V)
we say that this occurrence is
\emph{explicit}, if it comes from substring $ab$ of $\mathcal A$ (or $\mathcal B$, respectively);
\emph{implicit}, if it comes (wholly) from \sol X;
\emph{crossing} otherwise.
A pair is \emph{crossing} if it has a crossing occurrence
and \emph{non-crossing} otherwise.
Similar notion applies to maximal blocks of $a$s,
in which case we say that $a$ \emph{has a crossing block} or it \emph{has no crossing blocks}.
Alternatively, a pair $ab$ is crossing if $b$ is the first letter of \sol X and $aX$ occurs in the equation
or $a$ is the last letter of \sol X and $Xb$ occurs in the equation or $a$ is the last and $b$ the first letter of \sol X
and $XX$ occurs in the equation.

Unless explicitly stated, we consider crossing/non-crossing pairs $ab$ for $a \neq b$.
Note that as the first (last) letter of \sol X is the same for each \solution, see Lemma~\ref{lem:a prefix},
the definition of the crossing pair \emph{does not depend on the solution};
the same applies to crossing blocks.

When a pair $ab$ is non-crossing, its compression is easy, as it is enough to replace each
explicit $ab$ with a fresh letter $c$

\begin{algorithm}[H]
  \caption{$\algpair(a,b)$ Pair compression for a non-crossing pair \label{alg:pc}}
  \begin{algorithmic}[1]
  	\State let $c \in \letters$ be an unused letter
  	\State replace each explicit $ab$ in $\mathcal A$ and $\mathcal B$ by $c$
 \end{algorithmic}
\end{algorithm}

Similarly when none block of $a$ has a crossing occurrence,
the $a$'s blocks compression consists simply of replacing explicit $a$ blocks.

\begin{algorithm}[H]
  \caption{$\algblocks(a)$ Block compression for a letter $a$ with no crossing block
  \label{alg:ac}}
  \begin{algorithmic}[1]
  \For{each explicit $a$'s $\ell$-block occurring in $U$ or $V$ with $\ell>1$}
  		\State let $a_\ell \in \letters$ be an unused letter
		\State replace every explicit $a$'s $\ell$-block occurring
		in $\mathcal A$ or $\mathcal B$ by $a_\ell$
  \EndFor
  \end{algorithmic}
\end{algorithm}

\begin{lemma}
\label{lem:paircomp blockcomp}
Let $ab$ be a non-crossing pair then $\algpair(a,b)$  implements the pair compression for $ab$.
Let $a$ has no crossing blocks, then $\algblocks(a)$  implements the block compression for $a$.
\end{lemma}
\begin{proof}
Consider first the case of \algpair.
Suppose that $\mathcal A = {\mathcal B}$ has a solution \solution.
Define $\solution'$: $\solution'(X)$ is equal to \sol X with each $ab$ replaced with $c$
(where $c$ is a new letter).
Consider $\sol {\mathcal A}$ and $\solution'(\mathcal A')$. Then $\solution'(\mathcal A')$ is obtained from \sol{\mathcal A}
by replacing each $ab$ with $c$ (as $a \neq b$ this is well-defined):
the explicit occurrences of $ab$ are replaced by $\algpair(a,b)$,
the implicit ones are replaced by the definition of $\solution'$ and by the assumption
there are no crossing occurrences.
The same applies to $\sol {\mathcal B}$ and $\solution'(\mathcal B')$,
hence $\solution'$ is a solution of $\mathcal A' = \mathcal B'$.

Since $c$ is a free letter, the \sol {\mathcal A} is obtained from $\solution'({\mathcal A}')$ by replacing each $c$ with $ab$,
the same applies to \sol X and $\solution'(X)$ as well as \sol {\mathcal B} and $\solution'(\mathcal B')$.
Hence $\sol {\mathcal A} = \PC_{c \to ab} (\solution'(\mathcal A)) = \PC_{c \to ab} (\solution'(\mathcal B)) = \sol {\mathcal B}$
and $\sol X = \PC_{c \to ab}(\solution'(X))$, as required by the definition of implementing the pair compression.

Lastly, for a solution $\solution'$ of $\mathcal A' = \mathcal B '$ take the corresponding \solution{}
defined as $\sol X = \PC_{c \to ab}(\solution'(X))$ (i.e.\ replacing each $c$ with $ab$ in $\solution'(X)$).
It can be easily shown that $\sol {\mathcal A} = \PC_{c \to ab}(\solution'(\mathcal A '))$
and $\sol {\mathcal B} = \PC_{c \to ab}(\solution'(\mathcal B '))$, thus \solution{} is a solution of $\mathcal A = \mathcal B$.

The proof for the block compression follows in the same way.
\end{proof}

The main idea of the recompression method is the way it deals with the crossing pairs:
imagine $ab$ is a crossing pair, this is because $\sol X = bw$ and $aX$
occurs in $\mathcal A = \mathcal B$ or $\sol X = wa$ and $bX$ occurs in it
(the remaining case, in which $\sol X = awb$ and $XX$ occurs in the equation is treated in the same way).
The cases are symmetric, so we deal only with the first one.
To `uncross' $ab$ in this case it is enough to `left-pop' $b$ from $X$:
replace each $X$ in the equation with $bX$ and implicitly change the solution to $\sol X = w$.
Note that before replacing $X$ with $aX$ we need to check, whether $\sol X = a$ is a solution,
as this solution cannot be represented in the new equation;
similar remark applies to replacing $X$ with $Xb$.

\begin{algorithm}[H]
  \caption{$\algpop(a,b)$ \label{alg:leftpop}}
  \begin{algorithmic}[1]
		\If{$b$ is the first letter of \sol X}  \label{guess first letter}
			\If{$\algtestsimple(b)$ returns $1$} \Comment{$\sol X = b$ is a solution}
				\State report solution $\sol X = b$ \label{first letter test solution}
			\EndIf
			\State replace each $X$ in $\mathcal A = \mathcal B$ by  $bX$ \label{leftpop}
			\Comment{Implicitly change $\sol X = bw$ to $\sol X = w$}
		\EndIf
		\If{$a$ is the last letter of \sol X}  \label{guess last letter}
			\If{$\algtestsimple(a)$ returns $1$} \Comment{$\sol X = a$ is a solution}
				\State report solution $\sol X = a$
			\EndIf
			\State replace each $X$ in $\mathcal A = \mathcal B$ by  $Xa$ \label{rightpop}
			\Comment{Implicitly change $\sol X = w'a$ to $\sol X = w'$}
		\EndIf
  \end{algorithmic}
\end{algorithm}

\begin{lemma}
\label{lem:pop preserves solutions}
$\algpop(a,b)$ preserves solutions and after its application the pair $ab$ is noncrossing.
\end{lemma}

Note that Lemma~\ref{lem:pop preserves solutions} justifies our earlier claim that without loss of generality we can assume that
none of $A_i$, $B_j$ is empty: at the beginning of the algorithm we can run $\algpop(a,b)$ once
for $a$ being the first and $B$ the last letter of \sol X. 
This ensures the claim and increases the size of the instance at most thrice.

\begin{proof}
It is easy to verify that a pair $ab$ is crossing if and only if one of the following situations occurs:
\begin{enumerate}[CP1]
	\item \label{CP1}$aX$ occurs in the equation and the first letter of $\sol X$ is $b$;
	\item \label{CP2}$Xb$ occurs in the equation and the last letter of $\sol X$ is $a$;
	\item \label{CP3}$XX$ occurs in the equation, the first letter of $\sol X$ is $b$ and the last $a$.
\end{enumerate}
Let $\mathcal A' = \mathcal B'$ be the obtained equation,
we show that $ab$ in $\mathcal A' = \mathcal B'$ is noncrossing.
Consider whether $X$ was replaced by $bX$ is line~\ref{leftpop}.
If not, then the first letter of \sol X and $\solution'(X)$ is not $b$,
so $ab$ cannot be crossing because of \CPref{1} nor \CPref{3}.
Suppose that $X$ was replaced with $bX$. Then to the left of each $X$ there is a letter which is not $a$,
so none of situations \CPref{1}, \CPref{3} occurs.

A similar analysis applied to the last letter of \sol X yields that \CPref{2} cannot happen and so $ab$ cannot be a crossing pair.

\algpop{} can be naturally divided into two parts, which correspond to the replacement of $X$ by $bX$
and the replacement of $X$ by $Xa$.
We show for the first one that it preserves solutions, the proof for the second one is identical.

If \sol X does not begin with $b$ (recall that all solutions have the same first letter)
then nothing changes and the set of solutions is preserved.
Otherwise $\sol X = bw$. In this case the solutions of the new equation shall be obtained by prepending $b$ to them.
Consider what happens with a solution $\sol X = bw$
\begin{description}
	\item[if $w = \epsilon$] then it is reported in line~\ref{first letter test solution};
	\item[if $w \neq \epsilon$] then $\solution'(X) = w$ is a solution of the obtained equation.
\end{description}
Note that the solution reported by \algpop{} is verified, so it is indeed a solution.
Furthermore, the only reported solutions is $\sol X = b$,
none of which corresponds to a non-empty solution after popping.
Lastly, when $\solution'(X) = w$ is a solution after popping $b$ then clearly $\sol X = bw$
is a solution of $\mathcal A = \mathcal B$.
A symmetric analysis is done for the operation of right-popping $a$, which ends the proof.
\qedhere
\end{proof}

Now the presented procedures can be merged into one procedure
that turns crossing pairs into noncrossing ones and then compresses them,
effectively compressing crossing pairs.

\begin{algorithm}[H]
  \caption{$\algpairc(a,b)$ Turning crossing pair $ab$ into non-crossing ones and compressing it \label{alg:paircompc}}
  \begin{algorithmic}[1]
  			\State run $\algpop(a,b)$
			\State run $\algpair(a,b)$ \label{crossing pair compression}
  \end{algorithmic}
\end{algorithm}

\begin{lemma}
\label{lem: crossing pairs preserve}
$\algpairc(a,b)$  implements the pair compression of the pair $ab$.
\end{lemma}
The proof follows by combining Lemma~\ref{lem:paircomp blockcomp} and \ref{lem:pop preserves solutions}.

There is one issue: the number of non-crossing pairs can be large, however, a simple preprocessing,
which basically applies \algpop, is enough to reduce the number of crossing pairs to $2$.

\begin{algorithm}[H]
  \caption{\algpreproc{} Ensures that there are at most $2$ crossing pairs\label{alg:preproc}}
  \begin{algorithmic}[1]
		\State let $a$, $b$ be the first and last letter of \sol X
		\State run $\algpop(a,b)$
  \end{algorithmic}
\end{algorithm}

\begin{lemma}
\label{lem:preproc preserves solutions}
$\algpreproc$ preserves solution and after its application there are at most two crossing pairs.
\end{lemma}
\begin{proof}
It is enough to show that there are at most $2$ crossing pairs, as the rest follows form Lemma~\ref{lem:pop preserves solutions}.
Let $a$ and $b$ be the first and last letters of \sol X, and $a'$, $b'$ such letters after the application of \algpreproc.
Then each $X$ is proceeded with $a$ and succeeded with $b$ in $\mathcal A' = \mathcal B'$.
So the only crossing pairs are $aa'$ and $b'b$ (note that this might be the same pair or part of a letter-block, i.e. $a = a'$ or $b = b'$).
\qedhere
\end{proof}

The problems with crossing blocks can be solved in a similar fashion:
$a$ has a crossing block if and only if $aa$ is a crossing pair.
So we `left-pop' $a$ from $X$ until the first letter of \sol X is different than $a$,
we do the same with the ending letter $b$.
This can be alternatively seen as removing the whole $a$-prefix ($b$-suffix, respectively) from $X$:
suppose that $\sol X = a^\ell w b^r$, where $w$ does not begin with $a$ nor end with $b$.
Then we replace each $X$ by $a^\ell X b^r$ implicitly changing the solution to $\sol X = w$, see Algorithm~\ref{alg:prefix}.
\begin{algorithm}[H]
  \caption{\algprefsuff{} Cutting prefixes and suffixes; assumes that $A_0$ is not a block of letters\label{alg:prefix}}
  \begin{algorithmic}[1]
	\Require $A_0$ is not a block of letters, the non-empty of $A_{n_{\mathcal A}}$, $B_{n_{\mathcal B}}$ is not a block of letters 
	\State let $a$ be the first letter of \sol X
	\State report solution found by $\algtestsimple(a)$ \label{test prefix} \Comment{Excludes $\sol X \in a^+$ from further considerations.}
	\State let $\ell > 0$ be the length of the $a$-prefix of $A_0$ \Comment{By Lemma~\ref{lem:a prefix} \sol X has the same $a$-prefix} 
	\State replace each $X$ in $\mathcal A = \mathcal B$ by  $a^\ell X$
		\Comment{$a^\ell$ is stored in a compressed form},
	\State		\Comment{implicitly change $\sol X = a^\ell w$ to $\sol X = w$}
	\State let $b$ be the last letter of \sol X \Comment{It cannot be that $\sol X \in b^+$}
	\State report solution found by  $\algtestsimple(b)$ \Comment{Exclude $\sol X \in b^+$ from furhter considerations.}
	\State let $r > 0$ be the length of the $b$-suffix of the non-empty of $A_{n_{\mathcal A}}$, $B_{n_{\mathcal B}}$
		\par \Comment{By Lemma~\ref{lem:a prefix} \sol X has the same $b$-suffix} 
	\State replace each $X$ in $\mathcal A = \mathcal B$ by  $Xb^r$
		\Comment{$b^r$ is stored in a compressed form},
	\State		\Comment{implicitly change $\sol X = wb^r$ to $\sol X = w$}
  \end{algorithmic}
\end{algorithm}

\begin{lemma}
\label{lem:cutpref cutsuff}
Let $a$ be the first letter of the first word and $b$ the last of the last word.
If the first word is not a block of $a$s and the last not a block of $b$s then
\algprefsuff{} preserves solutions and after its application there are no crossing blocks of letters.
\end{lemma}
\begin{proof}
Consider first only the changes done by the modification of the prefix.
Suppose that $\sol X = a ^\ell w$, where $w$ does not begin with $a$.
If $w = \epsilon$ then as $A_0 \notin a^+$ by Lemma~\ref{lem:a prefix} there is only one such solution
and it is reported in line~\ref{test prefix}.
Otherwise, by Lemma~\ref{lem:a prefix}, each solution \solution{} of the equation
is of the form $\sol X = a^\ell w$, where $a^\ell$ is the $a$-prefix of $A_0$ and $w \neq \epsilon$ nor does it begin with $a$.
Then the $\solution'(X) = w$ is the solution of the new equation.
Similarly, for any solution $\solution'(X) = w$ the $\sol X = a^\ell w$
is the solution of the original equation.

The same analysis can be applied to the modifications of the suffix: observe that if at the beginning the last word
was not a block of $b$s it did not become one during the cutting of the $a$-prefix.

Lastly, suppose that some letter $c$ has a crossing block, without loss of generality assume that $c$ is the first
letter of \sol X and $cX$ occurs in the equation.
But this is not possible: $X$ was replaced by $a^\ell X$ and so the only letter to the left of $X$ is $a$
and \sol X does not start with $a$, contradiction.
\qedhere
\end{proof}

The \algprefsuff{} allows defining a procedure \algblocksc{}
that compresses maximal blocks of all letters, regardless of whether they have crossing blocks or not.

\begin{algorithm}[H]
	\caption{\algblocksc{} Compressing blocks of $a$}
	\label{alg:blocksc}
	\begin{algorithmic}[1]
		\State $\presentletters \gets$ letters occurring in the equation
		\State run \algprefsuff \Comment{Removes crossing blocks of $a$} \label{cut pref}
		\For{each letter $a \in \presentletters$} \label{loop of compressions}
			\State $\algblocks(a)$ \label{block compression local}
		\EndFor
	\end{algorithmic}
\end{algorithm}

\begin{lemma}
\label{lem: consistent no crossing block}
Let $a$ be the first letter of the first word and $b$ the last of the last word.
If the first word is not a block of $a$s and the last not a block of $b$s then
\algblocksc{}  implements the block compression for letters present in $\mathcal A = \mathcal B$
before its application.
\end{lemma}
The proof follows by combining Lemma~\ref{lem:paircomp blockcomp} and \ref{lem:cutpref cutsuff}.

\section{Main algorithm}
\label{sec:main algorithm}

The following algorithm \algsonevar{} is basically a specification
of the general algorithm for testing the satisfiability of word equations~\cite{wordequations}
and is built up from procedures presented in the previous section.

\begin{algorithm}[H]
	\caption{\algsonevar{} Reports solutions of a given one-variable word equation}
	\label{alg:onevar}
	\begin{algorithmic}[1]
	\While{the first block and the last block are not blocks of a letter} \label{alg:one mainloop}
		\State $\pairs \gets $ pairs occurring in $\sol {\mathcal A} = \sol {\mathcal B}$ \label{listing pairs}
		\State \algblocksc{} \Comment{Compress blocks, in $\Ocomp(|\mathcal A| + |\mathcal B|)$ time.}	\label{a* solution pref}	 
		\State \algpreproc{} \Comment{There are only two crossing pairs, see Lemma~\ref{lem:preproc preserves solutions}} \label{preproc}
		\State $\crossing \gets $ list of crossing pairs from \pairs
			\Comment{There are two such pairs} \label{listing crossing pairs} 
		\State $\ncrossing \gets $ list of non-crossing pairs from \pairs \label{listing noncrossing pairs}
		\For{each $ab\in \ncrossing$} \Comment{Compress non-crossing pairs, in time $\Ocomp(|\mathcal A| + |\mathcal B)|$}
			\State $\algpair(a,b)$ \label{pair compression onevar}
		\EndFor
		\For{$ab \in \crossing$} \label{loop of outer onever} \Comment{Compress the $2$ crossing pairs, in time $\Ocomp(|\mathcal A| + |\mathcal B)|$}
			\State $\algpairc(a,b)$ \label{pair compression onevar 2}
		\EndFor
	\EndWhile
	\State \algtest{} \Comment{Test solutions from $a^*$, see Lemma~\ref{lem:a* solution}}
 \end{algorithmic}
\end{algorithm}
We call one iteration of the main loop a \emph{phase}.

\begin{theorem}
\label{thm:onevar}
\algsonevar{} runs in time $\Ocomp(|{\mathcal A}| + |{\mathcal B}| + (n_{\mathcal A} + n_{\mathcal B})\log(|\mathcal A|+|\mathcal B|))$
and correctly reports all solution of a word equation $\mathcal A = \mathcal B$.
\end{theorem}

Before showing the running time, let us first comment on how the equation is stored.
Each of sides ($\mathcal A$ and $\mathcal B$) is represented as two lists of pointers to strings,
i.e.\ to $A_0$, $A_1$, \ldots, $A_{n_{\mathcal A}}$ and to $B_0$, $B_1$, \ldots, $B_{n_{\mathcal B}}$.
Each of those words is stored as a doubly-linked list.
When we want to refer to a concrete word in a phase, we use names $A_i$ and $B_j$,
when we want to stress its evolution in phases, we use names $\mathcal A$ $i$-word and $\mathcal B$ $j$-word.

\subsubsection*{Shortening of the solutions}
The most important property of \algsonevar{} is that the explicit strings between the variables shorten
(assuming that they have a large enough length).
To show this we use the following technical lemma, which is also used several times later on:

\begin{lemma}
\label{lem: shortening}
Consider two consecutive letters $a$, $b$ in \sol {\mathcal A} for any solution \solution.
At least one of those letters is compressed in this phase.
\end{lemma}
\begin{proof}
Consider whether $a = b$ or not:
\begin{description}
	\item[$a = b$] Then they are compressed using \algblocksc.
	\item[$a \neq b$]
	Then $ab$ is a pair occurring in the equation at the beginning of the phase and so it was listed in \pairs{} in line~\ref{listing pairs}
	and as such we try to compress it,
	either in line~\ref{pair compression onevar} or in line~\ref{pair compression onevar 2}.
	This occurrence cannot be compressed only when one of the letters $a$, $b$ was already compressed, in some other pair
	or by \algblocksc. In either case we are done.
	\qedhere
\end{description}
\end{proof}

We say that a word $A_i$ ($B_i$) is \emph{short} if it consists of at most $100$ letters and \emph{long} otherwise.
To avoid usage of strange constants and its multiplicities, we shall use $\bound=100$ to denote this value.
\begin{lemma}
\label{lem:reducing length}
Consider the length of the $\mathcal A$ $i$-word (or $\mathcal B$ $j$-word).
If it is long then its length is reduced by $1/4$ in this phase.
If it is short then after the phase it still is.
The length of each unreported solution is reduced by at least $1/4$ in a phase.

Additionally, if the first (last) word is short and has at least $2$ letters then its length is shortened by at least $1$ in a phase.
\end{lemma}
\begin{proof}
We shall first deal with the words and then comment how this argument extends to the solutions.
Consider two consecutive letters $a$, $b$ in any word at the beginning of a phase.
By Lemma~\ref{lem: shortening} at least one of those letters is compressed in this phase.
Hence each uncompressed letter in a word (the last letter) can be associated with the two letters
to the right that are compressed.
This means that in a word of length $k$ during the phase at least $\frac{2(k-1)}{3}$ letters are compressed
i.e.\ its length is reduced by at least $\frac{k-1}{3}$ letters.

On the other hand, letters are introduced into words by popping them from variables.
Let \emph{symbol} denote a single letter or block $a^\ell$ that is popped into a word,
we investigate, how many symbols are introduced in this way in one phase.
At most one symbol is popped to the left 
and one to the right by \algblocksc{} in line~\ref{a* solution pref},
the same holds for \algpreproc{} in line~\ref{preproc}.
Moreover, one symbol is popped to the left
and one to the right in line~\ref{pair compression onevar 2};
since this line is executed twice, this yields $8$ symbols in total.
Note that the symbols popped by \algblocksc{} are replaced by single letters,
so the claim in fact holds for letters as well.

So, consider any word $A_i \in \letters^*$ (the proof for $B_j$ is the same),
at the beginning of the phase and let $A_i'$ be the corresponding word at the
end of the phase.
There were at most $8$ symbols introduced into $A_i'$
(some of them might be compressed later).
On the other hand, by Lemma~\ref{lem: shortening}, at least $\frac{|A_i|-1}{3}$ letters were removed $A_i$ due to compression.
Hence
$$
|A_i'| \leq |A_i| - \frac{|A_i|-1}{3} + 8 \leq \frac{2|A_i|}{3} + 8\frac{1}{3} \enspace.
$$
It is easy to check that when $A_i$ is short, i.e.\ $|A_i| \leq \bound = 100$, then $A_i'$ is short as well
and when $A_i$ is long, i.e.\ $|A_i| > \bound$ then $|A_i'| \leq \frac{3}{4}|A_i|$.

It is left to show that the first word shortens by at least one letter in each phase.
Consider that if a letter $a$ is left-popped from $X$
then we created $B_0$ and in order to preserve~\eqref{eq:univariate} the first letters of $B_0$ and $A_0$ are removed.
Thus, $A_0$ gained one letter on the right and lost one on the left, so its length stayed the same.
Furthermore the right-popping does not affect the first word at all (as $X$ is not to its left);
the same analysis applies to cutting the prefixes and suffixes.
Hence the length of the first word is never increased by popping letters.
Moreover, if at least one compression (be it block compression or pair compression) is performed inside the first word, its length drops.
So consider the first word at the end of the phase let it be $A_0$.
Note that there is no letter representing a compressed pair or block in $A_0$:
consider for the sake of contradiction the first such letter that occurred in the first word.
It could not occur through a compression inside the first word (as we assumed that it did not happen),
cutting prefixes does not introduce compressed letters, nor does popping letters.
So in $A_0$ there are no compressed letters. But if $|A_0| > 1$ then this contradicts Lemma~\ref{lem: shortening}.

Now, consider a solution \sol X. We know that \sol X is either a prefix of $A_0$ or of the form $A_0^\ell A$, where
$A$ is a prefix of $A_0$, see Lemma~\ref{lem: solution form}.
In the former case, \sol X is compressed as a substring of $A_0$.
In the latter observe that argument follows in the same way, as long as we try to compress every pair of letters in \sol X.
So consider such a pair $ab$. If it is inside $A_0$ then we are done. Otherwise, $a$ is the last letter of $A_0$ and $b$ the first.
Then this pair occurs also on the crossing between $A_0$ and $X$ in $\mathcal A$, i.e.\ $ab$ is one of the crossing pairs.
In particular, we try to compress it.
So, the claim of the lemma holds for \sol X as well.
\end{proof}

The correctness of the algorithm follows from Lemmata~\ref{lem: consistent no crossing block}
(for \algblocksc), Lemma~\ref{lem:preproc preserves solutions} (for \algpreproc),
Lemma~\ref{lem:paircomp blockcomp} (for \algpair), Lemma~\ref{lem: crossing pairs preserve}
(for \algpairc)
and from the lemma below, which deals with \algtest.

\begin{lemma}
\label{lem:a* solution}
For $a \in \letters$ we can report all solutions in which $\sol X = a^\ell$
for some natural $\ell$ in $\Ocomp(|\mathcal A| +| \mathcal B|)$ time.
There is either exactly one $\ell$ for which $\sol X = a^\ell$ is a solution or
$\sol X = a^\ell$ is a solution for each $\ell$ or there is no solution of this form.
\end{lemma}
Note that we do not assume that the first or last word is a block of $a$s.
\begin{proof}
The algorithm and proof is similar as in Lemma~\ref{lem:a prefix}.
Consider a substitution $\sol X = a^\ell$.
We calculate the length of the $a$-prefix of $\sol {\mathcal A}$ and $\sol {\mathcal B}$.
Consider first letter other than $a$ in $\mathcal A$,
let it be in the $A_{k_A}$ and suppose that there were $\ell_A$ letters $a$ before it
(if there is non such letter, imagine we attach an `ending marker' to both $\mathcal A$ and $\mathcal B$,
which then becomes such letter).
Then the length of the $a$-prefix of \sol {\mathcal A} is $k_A \cdot \ell + \ell_A$.
Let additionally $\mathcal A'$ be obtained from $\mathcal A$ by removing those letters $a$ and variables in between them.
Similarly, define $k_B$, $\ell_B$ and $\mathcal B'$.
Then the length of the $a$-prefix of $\sol {\mathcal B}$ is $k_B \cdot \ell + \ell_B$.

The substitution $\sol X = a^\ell$ is a solution if and only if $k_A \cdot \ell + \ell_A = k_B \cdot \ell + \ell_B$
and $\sol {\mathcal A'} = \sol {\mathcal B'}$.
Consider the number of natural solutions of the equation
$$
k_A \cdot x + \ell_A = k_B \cdot x + \ell_B:
$$
\begin{description}
	\item[no natural solution]
	clearly there is no solution of the word equation $\mathcal A = \mathcal B$;
	\item[one solution $x = \ell$]
	then $\sol X = a^\ell$ is the only possible solution from $a^+$ of $\mathcal A = \mathcal B$.
	To verify whether \solution{} satisfies $\mathcal A' = \mathcal B'$
	we apply the same strategy as in $\algtestsimple(a)$: we evaluate both sides of $\mathcal A' = \mathcal B'$
	under the substitution $\sol X = a^\ell$ on the fly. The same argument as in Lemma~\ref{lem:a prefix} shows
	that the running time is linear in $|\mathcal A'| + |\mathcal B'|$
	\item[satisfied by all natural numbers]
	then the $a$-prefixes of $\mathcal A$ and $\mathcal B$ are of the same length for each $\sol X \in  a^*$.
	We thus repeat the procedure for $\mathcal A' = \mathcal B'$,
	shortening them so that they obey the form~\eqref{eq:univariate}, if needed.
	Clearly, solutions in $a^*$ of $\mathcal A' = \mathcal B'$ are exactly the solutions
	of $\mathcal A = \mathcal B$ in $a^*$.
\end{description}
The stopping condition for the recurrence above is obvious: if $\mathcal A'$ and $\mathcal B'$ are both empty then we are done
(each $\sol X = a^\ell$ is a solution of this equation),
if exactly one of them is empty and the other is not then there is no solution at all.

Lastly, observe that the cost of the subprocedure above is proportional to the amount of read letters, which are
then not read again, so the running time is $\Ocomp(|\mathcal A| + |\mathcal B|)$
\qedhere
\end{proof}

\subsubsection*{Running time}
Concerning the running time, we first show that one phase runs in linear time, which follows by standard approach,
and then that in total the running time is $\Ocomp(n + \#_X \log n)$.
To this end we assign in a fixed phase to each $\mathcal A$ $i$ word and $\mathcal B$ $j$ word 
cost proportional to its length.
For a fixed $\mathcal A$ $i$ word the sum of costs assigned while it was long forms a geometric sequence,
so sums up to at most constant more than the initial length of $\mathcal A$ $i$ word;
on the other hand the cost assigned when $\mathcal A$ $i$ word  is short is $\Ocomp(1)$ per phase ad there are $\Ocomp( \log n)$
phases.

\begin{lemma}
\label{lem:one iteration cost}
One phase of \algsonevar{} can be performed in $\Ocomp(|\mathcal A| + |\mathcal B|)$ time.
\end{lemma}
\begin{proof}
For grouping of pairs and blocks we use \algradix,
to this end it is needed that the alphabet of (used) letters can be identified
with consecutive numbers, i.e.\ with an interval of at most $|\mathcal A| + |\mathcal B|$ integers.
In the first phase of \algsonevar{} this follows from the assumption on the input.
\footnote{In fact, this assumption can be weakened a little: it is enough to assume that
$\letters \subseteq \{1,2,\ldots, \poly(|\mathcal A| + |\mathcal B|)\}$:
in such case we can use \algradix{} to sort \letters{} in
time $\Ocomp(|\mathcal A| +| \mathcal B|)$ and then replace \letters{} with set of consecutive natural numbers.}
At the end of this proof we describe how to bring back this property at the end of the phase.

To perform \algblocksc{} we want for each letter $a$ occurring in the equation
to have lists of all maximal $a$-blocks occurring in $\mathcal A = \mathcal B$
(note that after \algprefsuff{} there are no crossing blocks, see Lemma~\ref{lem:cutpref cutsuff}).
This is done by reading $\mathcal A = \mathcal B$ and listing triples $(a,k,p)$,
where $k$ is the length of a maximal block of $a$s and $p$ is a pointer to the beginning of this occurrence.
Notice, that the maximal block of $a$'s may consist also
of prefixes/suffixes that were cut from $X$ by \algprefsuff.
However, by Lemma~\ref{lem:a prefix} such a prefix is of length at most $|A_0| \leq |\mathcal A| +| \mathcal B|$\
(and similar analysis applies for the a suffix).
Then each maximal block includes at most one such prefix and one such suffix
thus the length of the $a$ maximal block is at most $3(|\mathcal A| +| \mathcal B|)$.
Hence, the triples $(a,k,p)$ can be sorted by their first two coordinates
using \algradix{} in total time $\Ocomp(|\mathcal A| +| \mathcal B|)$.

After the sorting, we go through the list of maximal blocks.
For a fixed letter $a$, we use the pointers to localise
$a$'s blocks in the rules and we replace each of its maximal block 
of length $\ell>1$ by a fresh letter.
Since the blocks of $a$ are sorted, all blocks of the same length
are consecutive on the list, and replacing them by the same letter is easily done.

To compress all non-crossing pairs, i.e.\ to perform the loop in line~\ref{pair compression onevar},
we do a similar thing as for blocks:
we read both $\mathcal A$ and $\mathcal B$, whenever we read a pair $ab$ where $a \neq b$ and both $a$ and $b$
are not letters that replaced blocks during the blocks compression,
we add a triple $(a,b,p)$ to the temporary list,
where $p$ is a pointer to this position.
Then we sort all these pairs according to lexicographic order on first two coordinates,
we use \algradix{} for that.
Since in each phase we number the letters occurring in $\mathcal A = \mathcal B$ using consecutive numbers,
this can be done in time $\Ocomp(|\mathcal A| +| \mathcal B|)$.
The occurrences of the crossing pairs can be removed from the list:
by Lemma~\ref{lem:preproc preserves solutions} there are at most two crossing pairs and they can be easily established
(by looking at $A_0XA_1$).
So we read the sorted list of pairs occurrences and we remove from it the ones that correspond to a crossing pair.
Lastly, we go through this list and replaces pairs, as in the case of blocks.
Note that when we try to replace $ab$ it might be that this pair is no longer there as one of its letters was already replaced,
in such a case we do nothing. This situation is easy to identify: before replacing the pair we check whether it is indeed
$ab$ that we expect there, as we know $a$ and $b$, this is done in costant time.

We can compress each of the crossing pairs naively in $\Ocomp(|\mathcal A| +| \mathcal B|)$ time
by simply first applying the popping and then reading the equation form the left to the right
and replacing occurrences of this fixed pair.

It is left to describe, how to enumerate (with consecutive numbers) letters in \letters{} at the end of each phase.
Firstly notice that we can easily enumerate all letters introduced in this phase
and identify them (at the end of this phase) with $\{1, \ldots, m\}$, where $m$
is the number of introduced letters (note that none of them were removed during the \algsolveeq).
Next by the assumption the letters in \letters{} (from the beginning of this phase)
are already identified with a subset of $\{1, \ldots , |\mathcal A|+|\mathcal B|\}$,
we want to renumber them,
so that the subset of letters from \letters{} that are present at the end of the phase
is identified with $\{m+1, \ldots, m+m'\}$ for an appropriate $m'$.
To this end we read the equation, whever we spot a letter $a$ that was present at the beginning of the phase
we add a pair $(a,p)$ where $p$ is a pointer to this occurrence.
We sort the list in time $\Ocomp(|\mathcal A|+|\mathcal B|)$.
From this list we can obtain a list of present letters together with list of pointers to their occurrences in the equation.
Using those pointers the renumbering is easy to perform in $\Ocomp(|\mathcal A|+|\mathcal B|)$ time.

So the total running time is $\Ocomp(|\mathcal A|+|\mathcal B|)$.
\qedhere
\end{proof}

The amortisation, especially in the next section, is much easier to be shown when we know that both the first and last words are long.
This assumption is not restrictive, as as soon as one of them becomes short, the remaining running time of \algsolveeq{} is linear.

\begin{lemma}
\label{lem: first is short}
As soon as first or last word becomes short,
the rest of the running time of \algsonevar{} is $\Ocomp(n)$.
\end{lemma}
\begin{proof}
One phase takes $\Ocomp(|\mathcal A|+|\mathcal B|)$ time by Lemma~\ref{lem:one iteration cost}
(this is at most $\Ocomp(n)$ by Lemma~\ref{lem:reducing length})
and as Lemma~\ref{lem:reducing length} guarantees that both the first word and the last word 
are shortened by at least one letter in a phase, there will be at most $\bound = \Ocomp(1)$ many phases.
Lastly, Lemma~\ref{lem:a* solution} shows that \algtest{} also runs in $\Ocomp(n)$.
\qedhere
\end{proof}

So it remains to estimate the running time until one of the last or first word becomes short.

\begin{lemma}
\label{lem: till one is first}
The running time of \algsonevar{} till one of first or last word becomes short
is $\Ocomp(n + (n_{\mathcal A} + n_{\mathcal B}) \log n)$.
\end{lemma}
\begin{proof}
By Lemma~\ref{lem:one iteration cost} the time of one iteration of \algsonevar{}
is $\Ocomp(|\mathcal A| + |\mathcal B|)$.
We distribute the cost among the $\mathcal A$ words and $\mathcal B$ words:
we charge $\beta |A_i|$ to $\mathcal A$ $i$-word and $\beta |B_j|$ to $\mathcal B$ $j$-word, for appropriate positive $\beta$.
Fix $\mathcal A$ $i$-word, we separately estimate how much was charged to it when
it was a long and short word.
\begin{description}
	\item[long] Let $n_i$ be the initial length of $\mathcal A$ $i$-word.
	Then by Lemma~\ref{lem:reducing length} the length in the $(k+1)$-th phase it at most $(\frac{3}{4})^k n_i$
	and so these costs are at most $\beta n_i + \frac{3}{4} \beta n_i + (\frac{3}{4})^2 \beta n_i + \ldots \leq 4 \beta n_i$.
	\item[short] Since $\mathcal A$ $i$-word is short, its length is at most $\bound$,
	so we charge at most $\bound \beta$ to it.
	Notice, that there are $\Ocomp(\log n)$ iterations of the loop in total,
	as first word is of length at most $n$ and it shortens by $\frac{3}{4}$ in each iteration when it is long
	and we calculate only the cost when it is long.
	Hence we charge in this way $\Ocomp(\log n)$ times, so in total $\Ocomp(\log n)$.
\end{description}
Summing those costs over all phases over all words and phases yields $\Ocomp((n_{\mathcal A} + n_{\mathcal B})\log n)$.
	\qedhere
\end{proof}

\section{Heuristics and Better Analysis}
\label{sec:faster}
The intuition gained from the analysis in the previous section, especially in Lemma~\ref{lem: till one is first}
is that the main obstacle in obtaining the linear running
time is the necessity of dealing with short words,
as  the time spend on processing them is difficult to charge.
This applies to both the compression performed within the short words,
which does not guarantee any reduction in length, see Lemma~\ref{lem:reducing length}, and to testing of the candidate solutions,
which cannot be charged to the length decrease of the whole equation.

Observe that by Lemma~\ref{lem: first is short} as soon as the first or last word becomes short,
the remaining running time is linear.
Hence, in our improvements of the running time we can restrict ourselves to the case, in which the first and last word are long.

The improvement to linear running time is done by four improvements in algorithm analysis and employed data structures,
which are described in details in the following subsections:
\begin{description}
	\item[several equations] Instead of a single equation, we store a system of several equations
	and look for a solution of such a system.
	This allows removal of some words from the equations that always correspond to each other.
	and thus decreases the overall storing space and testing time.
	This is described in Section~\ref{subsec: several} and Section~\ref{subsec: storing}.
	
	\item[small solutions] We identify a class of particularly simple solutions, called \emph{small},
	and show that a solution is reported within $\Ocomp(1)$ phases from the moment when it became small.
	In several problematic cases of the analysis we are able to show that the solutions involved are small
	and so it is easier to charge the time spent on testing them.
	Section~\ref{subsec: small} is devoted to this issue.
	
	\item[storage] The storage is changed so that all words are represented
	by a structure of size proportional to the size of the \emph{long words}.
	In this way the storage space decreases by a constant factor in each phase and so the running time (except for testing)
	is linear.
	This is explained in Section~\ref{subsec: storing} 
	
	\item[testing] The testing procedure is modified, so that the time it spends on the short words is reduced.
	In particular, we improve the rough estimate that one \algtestsimple{} takes time proportional to the equation
	to an estimation that actually counts for each word whether it was included in the test or not.
	Section~\ref{subsec: testing} is devoted to this.
\end{description}

\subsection{Suffix arrays and lcp arrays}
We use a standard data structure for comparisons on strings:
a suffix array $\textsl{SA}[1 \twodots m]$ for a string $w[1 \twodots m]$ stores the $m$ non-trivial suffixes
of $w$, that is $w[m], w[m-1\twodots m], \ldots, w[1\twodots m]$ in (increasing) lexicographical order.
In other words, $\textsl{SA}[k] = p$ if and only if $w[p \ldots m]$ is the $k$-th suffix according
to the lexicographical order.
It is known that such an array can be constructed in $\Ocomp(m)$ time~\cite{suffixarrays}
assuming that \algradix{} is applicable to letters,
i.e.\ that they are integers from $\{1, 2, \ldots, m^c\}$ for some constant $c$.
We assume explicitly that this is the case in our problem.

Using a suffix array the equality testing for substrings of $w$
reduces to the \emph{longest common prefix} (lcp) query:
observe that $w[i \twodots i + k] = w[j \twodots j+k]$ if and only if the common prefix of
$w[i \twodots m]$ and $w[j \twodots m]$ is at least $k$.
The first step in constructing a data structure for answering such queries is the LCP array:
for each $i = 1, \ldots, m-1$ the $LCP[i]$ stores the length of the longest common prefix of $SA[i]$ and $SA[i+1]$.
Given a suffix array, the LCP array can be constructed in linear time~\cite{lcpsuffixarrays}, however,
the linear-time construction of suffix arrays can be in fact extended to return also the LCP array~\cite{suffixarrays}.

When the LCP array is supplied, the general longest prefix queries reduce to the range minimum queries:
the longest common prefix of $SA[i]$ and $SA[j]$ (for $i < j$) is the minimum among $LCP[i], \ldots, LCP[j-1]$,
and so it is enough to have a data structure that answers the queries about the minimum in the range 
in constant time.
Such data structures in general case are known and in case of LCP arrays even simpler construction were given~\cite{rmq}.
The construction time is linear and query time is $\Ocomp(1)$~\cite{rmq}.
Hence, after a linear preprocessing, we can calculate the length of the longest common prefix of two substrings of a given string in $\Ocomp(1)$ time.

\subsection{Several equations}
\label{subsec: several}
The improved analysis assumes that we do not store a single equation,
instead, we store several equations and look for substitutions that simultaneously satisfy all of them.
Hence we have a collection $\mathcal A_i = \mathcal B_i$ of equations, for $i=1, \ldots, m$,
each of them is of the form described by~\eqref{eq:univariate};
by $\mathcal A = \mathcal B$ we denote the whole system of those equations.
In particular, each of those equations specifies the first and last letter of the solution,
length of the $a$-prefix and suffix etc.,
exactly in the same way as it does for a single equation.
If there is a conflict, as two equations give different answers regarding the first/last letter or the length of the $a$-prefix or $b$-suffix,
then there is no solution at all.
Still, we do not check the consistency of all those answers,
instead, we use an arbitrary equation, say $\mathcal A_1 = \mathcal B_1$, to establish the first, last letter, etc.,
and as soon as we find out that there is a conflict, we stop the computation and terminate immediately.

The system of equations stored by \algsonevar{} is obtained by replacing
one equation $\mathcal A_i'\mathcal A_i'' = \mathcal B_i' \mathcal B_i''$
(where $\mathcal A_i',\mathcal A_i'', \mathcal B_i', \mathcal B_i'' \in (\letters \cup \{ X\})^*$)
with equivalent two equations $\mathcal A_i' = \mathcal B_i'$ and $\mathcal A_i'' = \mathcal B_i''$
(note that in general the latter two equation are not equivalent to the former one, however,
we perform the replacement only when they are; moreover, we need to trim them so that they satisfy the form~\eqref{eq:univariate}).

The described way of splitting the equations implies a natural order on the equations in the system:
when $\mathcal A_i'\mathcal A_i'' = \mathcal B_i' \mathcal B_i''$ is split to
$\mathcal A_i' = \mathcal B_i'$ and $\mathcal A_i'' = \mathcal B_i''$ then 
$\mathcal A_i' = \mathcal B_i'$ is before $\mathcal A_i'' = \mathcal B_i''$
(moreover, they are both before/after each equation before/after which $\mathcal A_i'\mathcal A_i'' = \mathcal B_i' \mathcal B_i''$ was).
This order is followed whenever we perform any operations on all words of the equations.
We store a list of all equations, in this order.

We store each of those equations in the same way as described for a single equation in the previous phase, i.e.\
for an equation $\mathcal A_i = \mathcal B_i$ we store a list of pointers to words on one side and list of pointers to words on the other side.
Additionally, the first word of $\mathcal A_i$ has a link to the last word of $\mathcal A_{i-1}$ and the last word of $\mathcal A_i$
similarly, the last word of $\mathcal A_i$ has a link to the first word of $\mathcal A_i$ and the first word of $\mathcal A_{i+1}$.
We also say that $A_i$ ($B_j$) is first or last if it is in any of the stored equations.

All operations on a single equation introduced in the previous sections
(popping letters, cutting prefixes and suffixes, pair compression, blocks compression)
generalise to a system of equations.
The running times are addressed in detail later on.
Concerning the properties, they are the same, we list those for which the generalisation or the proof are non-obvious:
\algpreproc{} should ensure that there are only two crossing pairs.
This is the case, as each $X$ in every equation is replaced by the same $aXb$ and $\sol X$ is the same for all equations,
which is the main fact used in the proof of Lemma~\ref{lem:preproc preserves solutions}.
Lemma~\ref{lem:reducing length} ensured that in each phase the length of the first and last word is decreased.
Currently the first words in each equation may be different,
however, the analysis in Lemma~\ref{lem:reducing length} applies to each of them.

\subsection{Small solutions}
\label{subsec: small}
We say that a word $w$ is \emph{almost periodic} with \emph{period size} $p$ and \emph{side size} $s$
if it can be represented as $w = w_1 w_2^\ell w_3$ (where $\ell$ is an arbitrary number),
where $|w_2| \leq p$ and $|w_1w_3| \leq s$; we often call $w_2$ the \emph{periodic part} of this factorisation.
(Note that several such representation may exist,
we use this notion for a particular representation that is clear from the context).
A substitution \solution{} is \emph{small}, if $\sol X = ( w )^k v$,
where $w$, $v$ are almost periodic with period and side sizes $\bound$.

The following theorem shows the main result of this section: if a solution is small,
then it is reported by \algsonevar{} within $\Ocomp(1)$ phases.

\begin{theorem}
\label{thm:solution with small height}
If \sol X is a small solution then \algsonevar{} reports it within $\Ocomp(1)$ phases.
\end{theorem}

We would like to note that the rest of the paper is independent from the proof of Theorem~\ref{thm:solution with small height},
so it might be skipped in reading.

Intuition is as follows:
observe first that in each phase we make \algpop{} and test whether $\sol X = a$, where $a$ is a single letter, is a solution.
Thus it is enough to show that a small
solution is reduced to one letter within $\Ocomp(1)$ phases.
To see this, consider first an almost periodic word, represented as $ w_1 w_2^\ell w_3$.
Ideally, all compressions performed in one phase of \algsonevar{} are done separately on $w_1$, $w_3$ and each $w_2$.
In this way we obtain a string $w_1'w_2'^\ell w_3'$ and from Lemma~\ref{lem: shortening} it follows that
$w_i'$ is shorter than $w_i$ by a constant fraction.
After $\Ocomp(\log |w_2|)$ steps we obtain a word $w_1'' w_2''^\ell w_3'' $ in which $w_2''$ is a single letter,
and so in this phase $w_2''^\ell$ is replaced with a single letter.
Then, since the length of $w_1'''w_3'''$ is at most $\bound$, after $\Ocomp(1)$ phases this is also reduced to a single letter.
Concerning the small solution, $w^kv$ we first make such an analysis for $w$, when it is reduced to a single letter (after $\Ocomp(1)$ phases)
after one additional phase $w^k = a^k$ is also reduced to one letter (by \algblocksc) and so the obtained string $a_kv'$
is a concatenation of two almost periodic strings.
Using the same analysis as above for each of them we obtain that it takes $\Ocomp(1)$ time to reduce them all to single letters.
Thus we have a $2$-letter string, which is reduced to a single letter within $2$ phases.

In reality we need to take into the account that some compression are made on the crossings of the considered strings,
however, we can alter the factorisation (into almost periodic words and almost periodic words into periodic part and rest)
of the string so that the result is almost as in the idealised case.

We say that for a substring $w$ of \sol X during one phase of \algsonevar{} the letters in $w$ are \emph{compressed independently},
if every compressed pair or block were either wholly within this $w$ or wholly outside this $w$
(in some sense this corresponds to the non-crossing compression).

The following lemma shows that given an almost periodic substring of \sol X with period size $p$ and side size $s$
we can find an alternative representation in which the period size is the same, side size increases (a bit) but each $w$ in $w^k$
in this new representation is compressed independently.
This shows that the intuition about shortening of almost periodic strings is almost precise ---
we can think that periodic part in almost periodic strings are compressed independently,
but we need to pay for that by an increase in the side size.

\begin{lemma}
\label{lem:small solution partition}
Consider almost periodic substring of \sol X with period size $p$ and side size $s$ represented as $w_1w_2^\ell w_3$,
where $w$ is not a block of single letter.
Then there is a representation of this string as $w_1'w_2'^{\ell'}w_3'$ such that
\begin{itemize}
	\item $\ell-2 \leq \ell' \leq \ell$
	\item $|w_2'| = |w_2|$ (and consequently $|w_1'| + |w_3'| \leq |w_1| + |w_3| + 2|w_2|$)
	\item the form of $w_2'$ depends solely on $w_2$ and does not depend on $w_1$, $w_3$
	(it does depend on the equation and on the order of blocks and pairs compressed by \algsonevar)
	\item the compression in one phase of \algsonevar{} compresses each $w'$ from $w'^{\ell'}$ independently.
\end{itemize}
In particular, this other representation has period size $p$ and side size $s + 2p$.
\end{lemma}
\begin{proof}
First of all, if $\ell \leq 2$ then we take $w_2' = \epsilon$,
$k' = 0$ and concatenate $w_2^k$ to $w_1$ to obtain $w_1'$ (and take $w_3 = w_3'$).
So in the following we consider the case in which $\ell > 2$ and set $\ell' = \ell-2$.

Let $w_2 = a^m z b^r$, where $a,b \in \letters$, $m ,r\geq 1$ and $z \in \letters^*$ does not start with $a$ nor it ends with $b$,
such a representation is possible, as $w_2$ is not a block of letters.
Then $w_1w_2^\ell w_3 = w_1 (a^m zb^r)^\ell w_3$. Since $w_1$ can end with $a$ and $w_2$ can begin with $b$,
we are interested in compressions within the middle $zb^r (a^m z b^r)^{\ell-2}a^m z$.
We first show that indeed there is a compression of a substring that is fully within
the $zb^r (a^m z b^r)^{\ell-2}a^m z$:
\begin{itemize}
	\item If $m  > 1$ or $r> 1$ then we compress the block $a^m$ or $b^r$.
	\item If $m = r = 1$, $a = b$ then $ab = aa$ is a block and it is compressed.
	\item If $m = r = 1$, $a \neq b$ and $z = \epsilon$ then this substring is $b (ab)^{\ell-2} a$.
	As $\ell > 2$ the pair $ab$ is listed by \algsonevar{} and we try to compress it. If we fail then it means that one of the letters was already compressed with a letter inside the considered string.
	\item If $m = r = 1$ and $a \neq b$ and $z \neq \epsilon$ then $ba$ is listed among the pairs and we try to compress the occurrence right after the first $z$.
	If we fail then it means that one of the letters was compressed with its neighbouring letter, which is also in the string.	
\end{itemize}
Consider the first substring that is compressed and it is wholly within $zb^r (a^m z b^r)^{\ell-2}a^m z$.
There are two cases: the compressed substring is a block of letters or it is a pair.
We give a detailed analysis in the latter case, the analysis in the former case is similar.

So, let the first pair compressed wholly within this fragment $zb^r (a^m z b^r)^{\ell-2}a^m z$ be $cd$, see Fig.~\ref{fig:factors} for an illustration.
We claim that all pairs $cd$ that occurred within this fragment at the beginning of the phase are compressed at this moment.
Assume for the sake of contradiction that this is not the case.
So this means that one of the letters, say $c$ was already compressed in some other compression performed earlier.
By the choice of the compressed pair (i.e.\ $cd$), this $c$ is compressed with a letter from outside of the fragment
$zb (a^m z b^r)^{\ell-2}az$, there are two possibilities:
\begin{description}
	\item[$c$ is the last letter of $zb^r (a^m z b^r)^{\ell-2}a^m z$]
	Observe that the letter succeeding $c$ is either $b$ or a letter representing a compressed pair/string.
	In the latter case we do not make a further compression, so it has to be $b$.
	This is a contradiction: each $c$ that is a last letter of $z$ was initially followed by $b$,
	and so in fact some compression of $cb$ (note that by our choice the last letter of $z$ was not $b$, and so $b \neq c$)
	was performed wholly within $zb^r (a^m z b^r)^{\ell-2}a^m z$
	and it was done earlier than the compression of $cd$, contradiction with the choice of $cd$.
	\item[$c$ is the first letter of $zb^r (a^m z b^r)^{\ell-2}a^m z$]
	The argument is symmetric, with $a$ preceding $c$ in this case.
\end{description}

\begin{figure}
	\centering
		\includegraphics{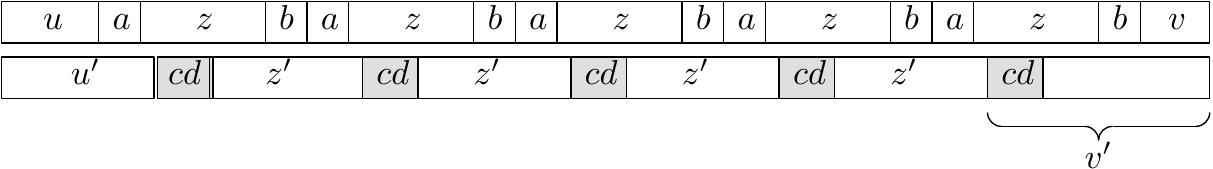}
	\caption{The alternative factorisation. The first compressed letters are in grey. For simplicity $m = r = 1$. Each $z'$ between the $cd$s is compressed independently.}
	\label{fig:factors}
\end{figure}

There are at least $\ell-1$ occurrences of $cd$ that are separated by $|w|-2$ letters,
i.e.\ the $(cdz')^{\ell-2}cd$ is a substring of $zb^r (a^m z b^r)^{\ell-2}a^m z$, for some $z'$ of length $|w|-2$,
see Fig.~\ref{fig:factors}.
We take $w_2' = cdz'$ and let $w_1'$ be the $w_1$ concatenated with string proceeding the $(cdz')^{\ell-2}cd$ and
$w_3'$ the $w_3$ concatenated with the string following this $(cdz')^{\ell-2}$  (note that the latter includes the ending $cd$, see Fig.~\ref{fig:factors}).
Clearly $|w_2'|= |w_2|$ and consequently $|w_1'| + |w_3'| = |w_1| + |w_3| + 2|w_2|$.
Note that each $w_2'$ begins with $cd$, which is the first substring even partially within $w_2'$ that is compressed,
furthermore, each of those $w_2'$ is also followed by $cd$. So the compression inside each $w_2'$ is done independently
(because by the choice of $cd$ there was no prior compression applied in $w_2'$).

Concerning the analysis when the first compressed substring is some $c^m$ it can be routinely verified that there are no essential differences in the analysis.
\qedhere
\end{proof}
The immediate consequence of Lemma~\ref{lem:small solution partition} is that when an almost periodic string is a substring of \solution,
then we can give bounds on the period size and side size on the corresponding word after one phase of \algsonevar.
\begin{lemma}
\label{lem: shortening of the almost periodic word}
Consider an almost periodic substring of \sol X with period size $p$ and side size $s$.
Then the corresponding substring after the phase of \algsonevar{} has a factorisation with period size at most $\frac{3}{4}p$
and side size at most $\frac{2}{3}s + \frac{7}{3}p$.
\end{lemma}
There are two remarks: firstly, if period size of the original word was $1$ then the given bound is $\frac{3}{4}<1$,
which holds, i.e.\ the corresponding word has no periodic part in the factorisation.
Secondly, the first (last) letter of the substring may be compressed with the letter to the left (tight, respectively),
so outside of the considered substring.
In such a case we still include the letter representing the replaced pair or block in the corresponding substring.
\begin{proof}
Let us fix the factorisation $w_1 w_2^\ell w_3$ of $w$,
where $p = |w_2|$ is the period size and $s = |w_1w_3|$ is the side size.
First of all, consider the special case, in which $w_2$ is a block of letters, say $a$,
without loss of generality we may assume that it is a single letter
(note that this simple borderline case is not covered by Lemma~\ref{lem:small solution partition}).
Without loss of generality we also may assume that $w_1$ does not end and $w_3$ does not begin with $a$,
as otherwise we can move those letters to $w_2^\ell$, decreasing the side size and not increasing the period size.
Then during the block compression the $w_2^\ell = a^\ell$ is  going to be replaced by a single letter
(this block may also include some letters from outside of $w$, when $w_1$ or $w_3$ is empty, this does not affect the analysis).
Now consider $w_1$: its first letter can be compressed with letters outside it,
otherwise each letter not compressed in the phase, except perhaps the last one, is followed by two letters that are, see Lemma~\ref{lem: shortening}.
Hence at most $2 + \frac{|w_1| - 2}{3}$ letters are uncompressed and so the length of the corresponding compressed $w_1'$
is at most $\frac{2|w_1| + 2}{3}$
and similarly for $w_3'$ its length is at most $\frac{2|w_3| + 2}{3}$.
Adding $1$ for the letter replacing $w_2^\ell$ we obtain $\frac{2|w_1w_3| + 7}{3} = \frac{2s}{3} + \frac{7p}{3}$, as claimed.

In other cases, by Lemma~\ref{lem:small solution partition} we can refactor $w$ into $u_1 u_2^{\ell'}u_3$
such that $|u_1u_3| \leq |w_1w_3| + 2|w_2|$ and $|u_2| = |w_2|$ and each $u_2$ is compressed independently (note that $|u_2| \geq 2$).
Then after one phase of \algsonevar{} the corresponding word $w'$ can be represented as $u_1'u_2'^{\ell'}u_3'$.
Let us inspect its compression rate.
The argument for $u_1$ and $u_3$ is the same as for $w_1$ and $w_2$ in the previous case,
so
$|u_1'| \leq \frac{2 |u_1| + 2}{3}$ and $|u_3'| \leq \frac{2 |u_3| + 2}{3}$. As $|u_1u_3| \leq |w_1w_3| + 2|w_2|$,
the new side size is at most $\frac{2}{3} s + \frac{4}{3} p + \frac{4}{3} \leq \frac{2}{3} s + 2p$, as $p \geq 2$.
For the period size, consider $u_2$. By Lemma~\ref{lem: shortening}, each uncompressed letter (perhaps except the last one)
is followed by a compressed one, and so $|u_2'| \leq \frac{2|u_2| + 1}{3}$.
For $|u_2| \geq 4$ this yields the desired compression rate $\frac{3}{4}$, for $|u_2| = 2$ and $|u_2| = 3$
observe that by Lemma~\ref{lem: shortening} at least one letter inside $u_2$ is compressed and we know that the compressions
inside $u_2$ are done independently, so $|u_2'| \leq |u_2| - 1$, which yields the desired bound for those two border cases.
\qedhere
\end{proof}

Imagine now we want to make a similar refactoring also for the small word
(in order to draw conclusions about shortening of \sol X, which is small).
So take $w^kv$ where both $w$ and $v$ are almost periodic words (with some period sizes and side sizes) and $k$ is some number.
When we look at $w^k$, each single $w$ can be refactored so that its periodic part is compressed independently.
Note that this is the same word, i.e.\ we are still given $w^kv$, though we have in mind a different factorisation of $w$.
However, the compression of the $w$ is influenced by the neighbouring letters,
so while each of the middle $w$ in $w^{k-2}$ is compressed in the same way,
both the first and the last $w$ can be compressed differently.
Hence, after the compression we obtain something of the form $w_1 w'^{k-2} w_2 v'$,
where $w_1, w', w_2, v'$ are almost periodic.
In the next phase the process continues and we accumulate almost periodic words on both sides of $w'^{k'}$.
So in general we deal with a word of the form $u w^k v$, where $w$ is almost periodic and $u,v$ are concatenations of almost periodic words.
The good news is that we can bound the sum of side sizes and period sizes of almost periodic words occurring in $u$, $v$.
Moreover, the period size of $w$ drops by a constant factor in each phase, so after $\Ocomp(1)$ phases it is reduced to $0$,
i.e.\ $w^k$ is almost periodic.

As a first technical step we show that Lemma~\ref{lem:small solution partition} can be used to analyse what happens with a concatenation of
almost periodic words in one phase of \algsonevar:
as in the case of a single word, see Lemma~\ref{lem: shortening of the almost periodic word},
the sum of period sizes drops by a constant factor, while the sum of side sizes drops by a constant factor but it increases by magnitude of
sum of period sizes.

\begin{lemma}
\label{lem: shortening of the almost periodic word concatenation}
Let $u$, a substring of \sol X, be a concatenation of almost periodic words with a factorisation for which the sum of period sizes if $p$
and side sizes is $s$.
Then after one phase of \algsonevar{} the corresponding string $u'$ is a concatenation of almost periodic words with a factorisation
for which the sum of period sizes is at most $\frac{3}{4} p$ and sum of side sizes is at most $\frac{2}{3}s + \frac{7}{3}p$.
\end{lemma}
Note that as in the case of Lemma~\ref{lem: shortening of the almost periodic word} when sum of the period sizes is $1$,
then after one phase we are guaranteed that all almost periodic words in the factorisation have empty periodic parts.
Moreover, as in the case of Lemma~\ref{lem: shortening of the almost periodic word} the first and last letter 
of $u$ may be compressed with the letters outside $u$, in which case we include in the corresponding word
the letters that are obtained in this way.
\begin{proof}
Let the promised factorisation of $u$ into almost periodic words be $u_1 \cdot u_2 \cdots u_m$.
We apply Lemma~\ref{lem: shortening of the almost periodic word} to each of them.
By a simple summation of the guarantees from Lemma~\ref{lem: shortening of the almost periodic word}
the bound on the size of the period sizes is $\frac{3}{4}p$ while the bound on the sum of the side sizes is $\frac{2}{3}s + \frac{7}{3}p$. 
\qedhere
\end{proof}

The following lemma is the crowning stone of our considerations. It gives bounds on the period sizes and side sizes for the word
that can be represented as $uw^kv$, where $w$ is almost periodic and $u, v$ are concatenations of almost periodic words.

\begin{lemma}
\label{lem:new sizes}
Suppose that at the beginning of the phase of \algsonevar{} 
a substring of \sol {\mathcal A_i} can be represented as $\sol X = u w^k v$, where
$w$ is almost periodic with period size $p_w>0$ and side size $s_w$ while $u,v$ are concatenations
of almost periodic words, let the sum of their period sizes be $p_{uv}$ and side sizes $s_{uv}$.
Then the corresponding substring at the end of the phase can be represented as $u' (w')^{k'} v'$,
where $w'$ is almost periodic with period size $p_w' \leq \frac{3}{4} p_w$
and side size $s_w' \leq \frac{2}{3}s_w + \frac{11}{3}p_w$ and $u', v'$ are concatenations of almost periodic words,
the sum of their period sizes is $p_{uv}' \leq \frac{3}{4}(p_{uv} + 2p_w)$ and the sum of their side sizes $s_{uv}'$
at most $\frac{2}{3}(s_{uv} + s_w) + \frac{7}{3}(p_{uv} + 2p_w)$.
\end{lemma}
\begin{proof}
Consider the factorisation of $w$ as an almost periodic word.
Consider first the main case, in which the periodic part of $w$ is not a block of single letter.
Then we can apply Lemma~\ref{lem:small solution partition} to each $w$,
obtaining a factorisation $w = w_1 w_2^\ell w_3$ such that $|w_2| \leq p_w$ and $|w_1w_3| \leq s_w + 2p_w$.
Then $uw^kv$ can be represented as
$$
u \left( w_1w_2^\ell w_3\right)^k v = u  w_1\left( w_2^\ell w_3 w_1 \right)^{k-1} w_2^\ell w_3 v \enspace .
$$
Define $u' = u w_1$ and $v' = (w_2^\ell) w_3 v$,
they are concatenations of almost periodic words, the sum of their period sizes is $p_{uv} + |w_2| = p_{uv} + p_v$
while side sizes $s_{uv} + |w_1| + |w_3| = s_{uv} + s_w + 2p_w$.
Define also $w' = w_2^\ell w_3 w_1$,
observe that each such $w'$ is delimited by $w_2$
(it includes it in the left end and to the right there is a copy of it which is not inside this $w'$)
and each $w_2$ is compressed independently, so also each $w'$ is compressed independently,
so in particular it is compressed in the same way.
Thus $u'w'^{k-1}v'$ is compressed into $u''w''^{k-1}v''$, let us estimate their sizes.

For $u'$ and $v'$ can can straightforwardly apply Lemma~\ref{lem: shortening of the almost periodic word concatenation},
obtaining that the sum of their period sizes is at most $\frac{3}{4}(p_{uv} + p_w)$
while their side sizes $\frac{2}{3}(s_{uv} + s_w + 2p_w) + \frac{7}{3}(p_{uv} + p_w) = \frac{2}{3}(s_{uv} + s_w) + \frac{7}{3}p_{uv} + \frac{11}{3}p_w$.
Concerning $w'$: we apply Lemma~\ref{lem: shortening of the almost periodic word}, which shows that the new period size is at most $\frac{3}{4}p_w$
and new side size $\frac{2}{3}(s_w + 2p_w) + \frac{7}{3}p_w = \frac{2}{3}s_w + \frac{11}{3}p_w$,
so all as claimed.

Let us return to the trivial case, in which $w = w_1 w_2^\ell w_3$ and $w_2$ is a block of a single letter.
Note that without a loss of generality, we can assume that $w_2$ is a single letter (we replace $w_2^\ell$ with $a^{|w_2|\ell}$)
and that $w_1$ does not end and $w_3$ does not begin with $a$
(we can move those letters to $w_2$, decreasing side size and not increasing the period size).
Then $uw^kv = u(w_1 a^\ell w_3)^kv$.
If $w_1w_3 = \epsilon$ this is equal to $u a^{k\ell} v$, we treat $a^{k\ell}$ as a almost periodic word with period size $1$ and side size $0$,
so $ua^{k\ell}v$ have a sum of period sizes $p_{uv} + 1 = p_{uv} + p_{w}$ and sum of side sizes $s_{uv}$.
Applying Lemma~\ref{lem: shortening of the almost periodic word concatenation} yields the claim:
the sum of period sizes is at most $\frac{3}{4}(p_{uv} + p_w)$ while the new side sizes $\frac{2}{3}s_{uv} + \frac{7}{3}(p_{uv} + p_w)$.
Similarly, when $k = 1$ we can treat $u w v$ as a concatenation of almost periodic words,
the sum of their period sizes is at most $p_{uv} + p_w$ and side sizes $s_{uv} + s_{w}$;
again, applying Lemma~\ref{lem: shortening of the almost periodic word concatenation} yields the claim:
the sum of period sizes is at most $\frac{3}{4}(p_{uv} + p_w)$ while the new side sizes $\frac{2}{3}(s_{uv}+s_w) + \frac{7}{3}(p_{uv} + p_w)$.

So let us go back to the main case, in which $w_1w_3 \neq \epsilon$ and $k \geq 2$.
Then $uw^kv = u(w_1 a^\ell w_3)^kv = uw_1 a^\ell w_3 w_1 (a^\ell w_3 w_1)^{k-2} a^\ell w_3 v$.
As $w_3w_1$ is non-empty and does not end, nor begin with $a$, each $a^\ell$ in $(a^\ell w_3 w_1)^{k-2}$ is compressed independently.
We set $u' = uw_1 a^\ell w_3 w_1$, $v' = a^\ell w_3 v$ and $w' = a^\ell w_3 w_1$.
Applying Lemma~\ref{lem: shortening of the almost periodic word concatenation} to $u'$ and $w'$
yields that after one phase the sum of period sizes is $\frac{3}{4}(p_{uv} + 2p_w)$
while side size $\frac{2}{3}(s_{uv} + 2s_w) + \frac{7}{3}(p_{uv} + 2p_w)$.
On the other hand, the period size of $w''$ is $\frac{3}{4}p_w$ while its side size at most $\frac{2}{3}s_w + \frac{7}{3}p_w$
\qedhere
\end{proof}

With Lemma~\ref{lem:new sizes} established, we can prove Theorem~\ref{thm:solution with small height}.

\begin{proof}[proof of Theorem~\ref{thm:solution with small height}]
Consider the string \sol X. We show that within $\Ocomp(\log \bound) = \Ocomp(1)$ this string is reduced to a single letter.
This means that \sol X is reported in the same time.
Note that in the following phases the corresponding solution (if unreported) \emph{is not} the corresponding string,
as we also pop letters from $X$. However, the corresponding solution is the substring of this string.

So fix a small solution and its occurrence within \sol {\mathcal A}.
It can be represented as $w^kv$, where $w$ and $v$ are almost periodic with period and side size $\bound$
We claim that in each following phase the corresponding string can be represented as
$u' w'^{k'} v'$, where $u'$ and $v'$ are concatenations of almost periodic words,
the sum of their period sizes is at most $6 \bound$ while side sizes $78 \bound$.
Also, $w'$ is almost periodic with side size at most $11 \bound$ and period size dropping by $\frac{3}{4}$ in each phase
(and at most $\bound$ at the beginning).
This claim can be easily verified by induction on the estimations given by Lemma~\ref{lem:new sizes}.
As the period size of $w'$ drops by $\frac{3}{4}$ in each phase and initially it is $\bound$,
after $\Ocomp(\log \bound)$ phases $w'$ has period size $0$.
Then inside $u' w'^{k'} v'$ we treat $w'^{k'}$ as a periodic word with period size $|w'| \leq 11 \bound$ and side size $0$.
Thus $u' w'^{k'} v'$ is a concatenation of almost periodic words with sum of period sizes at most $17 \bound$ and sum of side size at most $78 \bound$.
Then, by easy induction on bounds given by Lemma~\ref{lem: shortening of the almost periodic word concatenation},
in the following phases the corresponding string will be a concatenation of almost periodic strings,
with sum of period sizes decreasing by $\frac{3}{4}$ in each phase (and initial value $17 \bound) $and sum of side sizes at most $78 \bound$.
Thus after $\Ocomp(\log \bound )$ phase its sum of period sizes is reduced to $0$ and so it is a string of length at most $78 \bound$,
which will be reduced to a single letter within $\Ocomp(\log \bound)$ rounds, as claimed.
Since $\bound = \Ocomp(1)$.
\qedhere
\end{proof}

\subsection{Storing of an equation}
\label{subsec: storing}
To reduce the running time we store duplicates of short word only once.
Recall that for each equation we store lists of pointers pointing to strings that are the explicit words in this equation.
We store the long words in a natural way, i.e.\ each long word is represented by a separate string.
The short words are stored more efficiently:
if two short words in equations are equal we store only one string, to which both pointers point.
In this way all identical short words are stored only once (though each of them has a separate pointer pointing to it);
we call such a representation \emph{succinct}.

We show that the compression can be performed on the succinct representation,
without the need of reading the actual equation.
This allows bounding the running time using the size of the succinct representation and not the equation.

We distinguish two types of short words:
those that are substrings of long words (normal) and those that are not (overdue).
We can charge the cost of processing the normal short words to the time of processing the long words.
The overdue words can be removed from the equation after $\Ocomp(1)$ phases after becoming overdue,
so their processing time is constant per $\mathcal A$-$i$ word (or $\mathcal B$-$j$ word).

The rest of this subsection is organised as follows:
\begin{itemize}
	\item We first give precise details, how we store short and long words, see Section~\ref{subsubsec: storing details}
	and prove that we can perform compression using only succinct representation, see Lemma~\ref{lem:words are equal}.
	\item We then define precisely the normal and overdue words, see Section~\ref{subsubsec: overdue}
	as well as show that we can identify new short and overdue words, see Lemma~\ref{lem:identify overdue}.
	Then we show that overdue words can be removed $\Ocomp(1)$ phases after becoming overdue,
	see Lemma~\ref{lem:overdue can be removed} and \ref{lem:overdue}.
	\item Lastly, in Section~\ref{subsubsec: compression time}, we show that the whole compression time,
	summed over all phases is $\Ocomp(n)$.
	The analysis is done separately for long words normal short words and overdue short words.
\end{itemize}

As observed at the beginning of Section~\ref{sec:faster}, as soon as the first or last word becomes short, the remaining running time is linear.
Thus, when such a word becomes short,
we drop our succinct representation and recreate out of it the simple representation used in Sections\ref{sec:prelim}--\ref{sec:main algorithm}.
Such a recreation takes linear time.

\subsubsection{Storing details}
\label{subsubsec: storing details}
We give some more details about the storing:
All long words are stored on two doubly-linked lists, one representing the long words on the left-hand sides
and the other the long words on the right-hand sides.
Those words are stored on the lists according to the initial order of the words in the input equation.
Furthermore, for each long word we store additionally, whether it is a first or last word of some equation
(note that a short word cannot be first or last).
The short words are also organised as a list, the order on the list is irrelevant.
Each short word has a list of its occurrences in the equations, the list points to the occurrences in the natural order
(occurrences on the left-hand sides and on the right-hand sides are stored separately).

We say that such a representation is \emph{succinct} and its size is the sum of 
lengths of words stored in it
(so the sum of sizes of long words, perhaps with multiplicities, plus the sum of sizes of different short words).
Note that we do \emph{not} include the number of pointers from occurrences of short words.
We later show that in this way we do not need to actually read the whole equation in order to compress it;
it is enough to read the words in the succinct representation, see Lemma~\ref{lem:compression cost}.

We now show that such a storage makes sense, i.e.\ that if two short words become equal,
they remain equal in the following phases (note again that none of them are first, nor last).
\begin{lemma}
\label{lem:words are equal}
Consider any explicit words $A$ and $B$ in the input equation.
Suppose that during \algsonevar{} they were transformed to $A' = B'$,
none of which is a first or last word in one of the equations.
Then $A = B$ if and only if $A' = B'$.
\end{lemma}
\begin{proof}
By induction on operation performed by \algsonevar.
Since none of the $A'$, $B'$ is the first or last word in the equation,
it means that during the whole \algsonevar{} they had $X$ to the left and to the right.
So whenever a letter was left-popped or right-popped from $X$, it was prepended or appended to both
$A$ and $B$; the same applies to cutting prefixes and suffixes.
Compression is never applied to a crossing pair or a crossing block, so after it
two strings are equal if and only if they were before the operation.
The removal of letters (in order to preserve~\eqref{eq:univariate}) is applied only to first and last words, so it does not apply to words considered here.
Partitioning the equation into subequations does not affect the equality of explicit words.
\end{proof}

We now show the main property of succinct representation:
the compression (both pair and block) can be performed on succinct representation in linear time.

\begin{lemma}
\label{lem:compression cost}
The compression in one phase of \algsonevar{}
can be performed in time linear in size of the succinct representation.
\end{lemma}
\begin{proof}
The long words are stored in a list and we can compress them without the need of reading the word table.
We know which one of them is first or last, so when letters are popped from $X$ we know what letters are
appended/prepended to each of those words.
Since they are stored explicitly, the claim for them follows from the analysis of the original version of \algsonevar,
see Lemma~\ref{lem:one iteration cost}. This analysis in particular requires that we can identify
the letters used in the equation with numbers from an interval of linear size.
Here the size is the size of the succinct representation. Note though that this part of the proof
follows in the same way: when listing letters (to replace them with new ones) we do not need to list letters in different
occurrences of the same short word, it is enough to do this once, which can be done using the succinct representation.

For the short words stored in the list of short words,
from Lemma~\ref{lem:words are equal} it follows that if an explicit word $A$ occurs twice in the equations
(both times not as a first, nor last word of the equation) it is changed during \algsonevar{} in the same way
at both those instances.
So it is enough to perform the operations on the words stored in the list,
doing so as in the original version of \algsonevar{} takes time linear in the size of the tables of short words,
as in Lemma~\ref{lem:one iteration cost}.
\end{proof}

\subsubsection{Normal and overdue short words}
\label{subsubsec: overdue}
The short words stored in the tables are of two types: normal and overdue.
The \emph{normal} words are substrings of the long words or $A_0^2$ and consequently the sum of their sizes
is proportional to the size of the long words.
A word becomes \emph{overdue} if at the beginning of the phase it is not a substring of a long word nor $A_0^2$.
It might be that it becomes a substring of such a word later, it does not stop to be an overdue word in such a case.

Since the normal words are of size $\Ocomp(\bound) = \Ocomp(1)$, the sum of lengths of normal words stored in short word list
is at most $\Ocomp(1)$ larger than the sum of sizes of the long words.
Hence the processing time of normal short words can be charged to the long words.
For the overdue words the analysis is different:
we show that after $\Ocomp(1)$ phases we can remove them from the equation (splitting the equations).
Thus their processing time is $\Ocomp(1)$ per $\mathcal A$-$i$ word (or $\mathcal B$-$j$ word)
and thus $\Ocomp(n)$ in total.

The new overdue words can be identified in linear time:
this is done by constructing a suffix array 
for a concatenation of long and short words occurring in the equations.

\begin{lemma}
\label{lem:identify overdue}
In time proportional to the size of succinct representation size we can identify the new overdue words.
\end{lemma}
\begin{proof}
Consider all long words $A_0$, \ldots, $A_m$ (with or without multiplicities, it does not matter)
and all short (not already overdue) words $A_1'$, \ldots $A'_{m'}$, without multiplicities;
in both cases this is just a listing of words stored in the representation (except for old overdue words).
We construct a suffix array for the string
$$
A_0^2\$A_1\$\dots A_m\$A_1'\$\dots A'_{m'}\# \enspace .
$$
As it was already observed that the size of the alphabet is linear in the size of the succinct representation,
the construction of the suffix array can be done in linear time~\cite{suffixarrays}.

Now $A'_i$ is a factor in some $A_j$ (the case of $A_0^2$ is similar, it is omitted to streamline the presentation)
if and only if for some suffix $A_j''$ of $A_j$ the strings $A_j''\$A_{j+1}\dots A_m\$A_1'\$\dots \$A'_{m'}\#$
and $A_i'\$\dots \$ A'_{m'}\#$ have a common prefix of length at least $|A_i'|$.
In terms of a suffix array, the entries for $A_i'\$\dots \$ A'_{m'}\#$
and $A_j''\$A_{j+1}\dots \$ A_m\$A_1'\$\dots \$ A'_{m'}\#$ should have a common prefix of length at least $|A_i'|$.
Recall that the length of the longest common prefix of two suffixes stored at positions $p < p'$ in
the suffix array is the minimum of $LCP[p]$, $LCP[p+1]$, \ldots, $LCP[p'-1]$.

For fixed suffix $A_i'\$\dots \$ A'_{m'}\#$ we want to find $A_j''\$A_{j+1}\dots \$ A_m\$A_1'\$\dots \$ A'_{m'}\#$
(where $A_j''$ is a suffix of some long word $A_j$)
with which it has the longest common prefix.
As the length of the common prefix of $p$th and $p'$th entry in a suffix array is
$\min(LCP[p], LCP[p+1], \ldots, LCP[p'-1])$, this is is either the first previous or first next
suffix of this form in the suffix array.
Thus the appropriate computation can be done in linear time:
we first go down in the suffix array, storing the last spotted
entry corresponding to a suffix of some long $A_j$, calculating the LCP with consecutive suffixes
and storing them for the suffixes of the form $A_i'\$\dots \$ A'_{m'}\#$.
We then do the same going from the bottom of the suffix array.
Lastly, we choose the larger from two stored values;
for $A_i'\$\dots \$ A'_{m'}\#$ it is smaller than $|A_i'|$ if and only if $A_i'$ just became an overdue word.

Concerning the running time, it linearly depends on the size of the succinct representation and alphabet size,
which is also linear in size of succinct representation, as claimed.
\qedhere
\end{proof}

The main property of the overdue words is that they can be removed from the equations in $\Ocomp(1)$ phases after becoming overdue.
This is shown by a serious of lemmata.

First we need to define what does it mean that for solution word $A$ in one side of the equation is at the same position as its copy
on the other side of the equation:
we say that for a substitution \solution{} the explicit word $A_i$ (or its subword)
is \emph{arranged against} the explicit word $B_j$ (\sol X for some fixed occurrence of $X$) 
if the position within $\sol{\mathcal A_k}$ occupied by this explicit word $A_i$ (or its subword)
are within the positions occupied by explicit word $B_j$ (\sol X, respectively) in $\mathcal B_k$.

\begin{lemma}
\label{lem:overdue can be removed}
Consider a short word $A$ in a phase in which it becomes overdue.
Then for each solution \sol X either \solution{} is small
or in every $\sol {\mathcal A_k} = \sol {\mathcal B_k}$ each explicit word $A_i$ equal to $A$
is arranged against another explicit word $B_j$ equal to $A$.
\end{lemma}
\begin{proof}
Consider an equation and a solution \solution{} such that in some
$\sol {\mathcal A_i} = \sol {\mathcal B_i}$ an explicit word $A_i$ (equal to an overdue word $A$)
is not arranged against another explicit word equal to $A$.
There are three cases:

\subsubsection*{$A$ is arranged against \sol X}
Note that in this case $A$ is a substring of \sol X. 
Either \sol X is a substring of $A_0$ or $\sol X = A_0^kA_0'$, where $A_0'$ is a prefix of $A_0$.
In the former case $A$ is a factor of $A_0$, which is a contradiction,
in the latter it is a factor of $A_0^{k+1}$. As $A_0$ is long and $A$ short, it follows that $|A|<|A_0|$
and so $A$ is a factor of $A_0^2$, contradiction with the assumption that $A$ is overdue.

\begin{figure}
	\centering
		\includegraphics{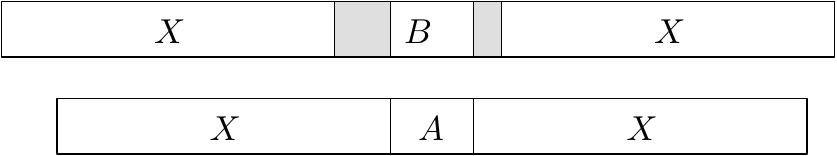}
	\caption{$A$ is arranged against $B$. The periods of length at most $|B| - |A|$ are in ligther grey.
	Since $A \neq B$, at least one of them is non-empty.}
	\label{fig:arranged}
\end{figure}

\subsubsection*{$A$ is arranged against some word}
Since $A$ is an overdue word, this means that $A_i$ is arranged against a short word $B_j$.
Note that both $A_i$ and $B_j$ are preceded and succeeded by \sol X,
since $A_i \neq B_j$ we conclude that \sol X has a period at most $|B_j| - |A_i|$, see Fig.~\ref{fig:arranged};
in particular \solution{} is small.

\begin{figure}
	\centering
		\includegraphics{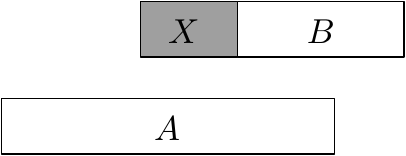}
	\caption{Subword of $A_i$ is arranged against the whole \sol X.}
	\label{fig:misarranged1}
\end{figure}

\subsubsection*{Other case}
Since $A_i$ is not arranged against any word, nor arranged against \sol X, it means that
some substring of $A_i$ is arranged against \sol X{} and as $A_i$ is preceded and succeeded by \sol X,
this means that either \sol X is shorter than $A_i$ or it has a period at most $|A|$, see Figure~\ref{fig:misarranged1} and~\ref{fig:misarranged2}, respectively.
In both cases \solution{} is small.
\begin{figure}
	\centering
		\includegraphics{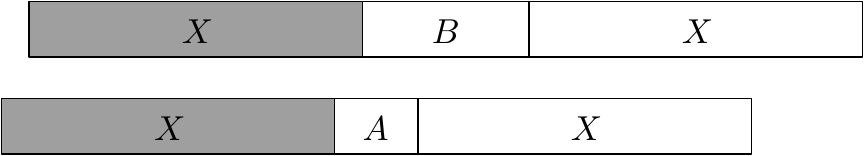}
	\caption{Subword of $A_i$ is arranged against \sol X.
	The overlapping \sol X are in in grey, the \sol X has a period shorter than $A_i$.}
	\label{fig:misarranged2}
\end{figure}
\qedhere
\end{proof}

Observe that due to Theorem~\ref{thm:solution with small height} and Lemma~\ref{lem:overdue can be removed}
the $\mathcal A$-$i$-words and $\mathcal B$-$j$-words
that are overdue can be removed in $\Ocomp(1)$ phases after becoming overdue:
suppose that $A$ becomes an overdue word in phase $\ell$.
Any solution, in which an overdue word $A$ is not arranged against another occurrence of $A$
is small and so it is reported after $\Ocomp(1)$ phases.
Consider an equation $\mathcal A_i = \mathcal B_i$ in which $A$ occurs.
Then the first occurrence of $A$ in $\mathcal A_i$ and the first occurrence of $A$ in $\mathcal B_j$
are arranged against each other for each solution \solution.
In particular, we can write $\mathcal A_i = \mathcal B_i$ as 
$\mathcal A_i' XAX \mathcal A_i'' = \mathcal B_i' XAX \mathcal B_i''$,
where $\mathcal A_i$ and $\mathcal B_i$ do not have $A$ as an explicit word
(recall that $A$ is not the first, nor the last word in $\mathcal A_i = \mathcal B_i$).
This equation is equivalent to two equations $\mathcal A_i'  = \mathcal B_i' $ and $\mathcal A_i'' = \mathcal B_i''$.
This procedure can be applied recursively to $\mathcal A_i'' = \mathcal B_i''$.
In this way, all occurrences of $A$ are removed and no solutions are lost in the process.
There may be many overdue strings so the process is a little more complicated,
however, as each word can be removed once during the whole algorithm,
in total it takes $\Ocomp(n)$ time.

\begin{lemma}
\label{lem:overdue}
Consider the set of overdue words introduced in phase $\ell$.
Then in phase $\ell + \Ocomp(1)$ we can remove all occurrences of overdue words from the equations.
The obtained set of equations has the same set of solutions.
The amortised time spend on removal of overdue words, over the whole run of \algsonevar, is $\Ocomp(\#_X)$.
\end{lemma}
\begin{proof}
Consider any word $A$ that become overdue in phase $\ell$ and any solution \solution{} of this equation,
such that in some $\sol {\mathcal A_i} = \sol {\mathcal B_i}$
the explicit word $A$ is not arranged against another instance of the same explicit word.
Then due to Lemma~\ref{lem:overdue can be removed} the \sol X is small.
Consequently, from Theorem~\ref{thm:solution with small height} this solution is reported before phase $\ell + c$, for some constant $c$.
So any solution $\solution'$ in phase $\ell + c$ corresponds to a solution \solution{} from phase $\ell$
that had each explicit word $A$ arranged in each $\sol {\mathcal A_i} = \sol {\mathcal B_i}$ against another explicit word $A$.
Since all operations in a phase either transform solution, implement the pair compression of implement
the blocks compression for a
solution \sol X,
it follows that in phase $\ell + c$ the corresponding overdue words $A'$
are arranged against each other in $\solution'(\mathcal A_i') = \solution'(\mathcal B_i')$.
Moreover, by Lemma~\ref{lem:words are equal} each explicit word $A'$ in this phase corresponds to an explicit word $A$ in phase $\ell$.

This observation allows removing all overdue words introduced in phase $\ell$.
Let $C_1$, $C_2$, \ldots, $C_m$ (in phase $\ell + c$) correspond to all overdue words introduced in phase $\ell$.
By Lemma~\ref{lem:identify overdue} we have already identified the overdue words.
Using the list of short words, for each overdue word $C$, we have
the list of pointers to occurrences of $C$ in left-hand sides of the equations and right-hand sides of the equations,
those lists are sorted according to the order of occurrences.
In phase $\ell + c$ we go through those lists, if the first occurrences of $A$ in the left-hand sides and right-hand sides
are in different equations then the equations are not satisfiable,
as this would contradict that in each solution both $A$ is arranged against its copy.
Otherwise, they are in the same equation $\mathcal A_i = \mathcal B_i$, which is of the form
$\mathcal A_i' XAX \mathcal A_i'' = \mathcal B_i' XAX \mathcal B_i''$,
where $\mathcal A_i'$ and $\mathcal B_i'$ do not have any occurrence of $A$ within them.
We split $\mathcal A_i = \mathcal B_i$ into two equations $\mathcal A_i'  = \mathcal B_i' $ and $\mathcal A_i'' = \mathcal B_i''$
and we trim them so that they are in the form described in~\eqref{eq:univariate}.
The new equations have exactly the same set of solutions as the original one.

Note that as new equations are created, we need to reorganise the pointers from the first/last words in the equations,
however, this is easily done in $\Ocomp(1)$ time.
The overall cost can be charge to the removed $X$, which makes in total at most $\Ocomp(\#_X)$ cost.
\qedhere
\end{proof}

\subsubsection{Compression running time}
\label{subsubsec: compression time}
\begin{lemma}
\label{lem:time for data}
The running time of \algsonevar, except for time used to test the solutions, is $\Ocomp(n)$.
\end{lemma}
\begin{proof}
By Lemma~\ref{lem:compression cost} the cost of compression is linear in terms of the size of the succinct representation
by Lemma~\ref{lem:identify overdue} in the same time bounds we can also identify the overdue words.
Lastly, by Lemma~\ref{lem:overdue can be removed} the total cost of removing the overdue words is $\Ocomp(n)$.
So it is enough to show that the sum of sizes of the succinct representations summed over all phases is $\Ocomp(n)$.

When the overdue words are excluded, the size of the succinct representation is proportional to the total length of long words.
Since by Lemma~\ref{lem:reducing length} this sum of lengths decreases by a constant in each phase,
the sum of those costs is linear in $n$.

Concerning the costs related to the overdue words: Note that an $\mathcal A$ $i$-word or $\mathcal B$ $j$-word
is overdue for only $\Ocomp(1)$ phases, after which it is deleted from the equation see Lemma~\ref{lem:overdue}.
So in $\Ocomp(1)$ phases it is charged $\Ocomp(\bound) = \Ocomp(1)$ cost, during the whole run of \algsonevar.
Summing over all $\mathcal A$ $i$-words and $\mathcal B$ $j$-words yields $\Ocomp(n)$ time.
\qedhere
\end{proof}

\subsection{Testing}
\label{subsec: testing}
We already know that thanks to appropriate storing the compression of the equations can be performed in linear time.
It remains to explain how to test the solutions fast, i.e.\ how to perform \algtestsimple{}
when all first and last words are still long.

Recall that \algtestsimple{} checks whether \solution,which is of the form $\sol X = a^\ell$ for some $\ell$,
is a solution by comparing \sol {\mathcal A_i} and \sol {\mathcal B_i} letter by letter, replacing $X$ with $a^\ell$ on the fly.
We say that in such a case a letter $b$ in \sol {\mathcal A_i} is \emph{tested against} the corresponding letter in \sol {\mathcal B_i}.
Note that during the testing we do not take advantage of the smaller size of the succinct representation,
so we need to make a separate analysis.
Consider two letters, from $A_i$ and $B_j$, that are tested against each other.
If one of $A_i$ and $B_j$ is long, this can be amortised against the length of the long word.
The same applies when one of the words $A_{i+1}$ or $B_{j+1}$ is long.
So the only problematic case is when all of those words are short.
To deal with this case efficiently we distinguish between different test types,
in which we exploit different properties of the solutions to speed up the tests.
In the end, we show that the total time spent on testing is linear.

For a substitution \solution{} by a \emph{mismatch} we denote the first position on which \solution{} is shown not be a solution,
i.e.\ sides of the equation have different letters (we use a natural order on the equations);
clearly, a solution has no mismatch.
Furthermore, \algsonevar{} stops the testing as soon as it finds a mismatch,
so in the rest of this section, if we use a name \emph{test} for a comparison of letters,
this means that the compared letters are before the mismatch (or that there is no mismatch at all).

There are two preliminary technical remarks:
First we note that for when testing a substitution \solution,
for a fixed occurrence of $X$ there is at most explicit word whose letters are tested against letters from this occurrence of $X$.

\begin{figure}
	\centering
		\includegraphics{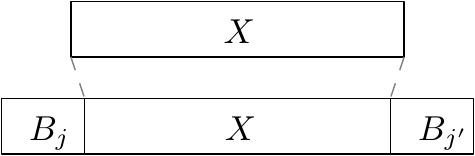}
	\caption{Let $B_{j}$ and $B_{j'}$ both have their letters arranged against letters from fixed occurrence of $X$.
	Then the $X$ separating them is a proper substring of another $X$, contradiction.\label{fig:protected}}
\end{figure}

\begin{lemma}
\label{lem: against X}
Fix a tested substitution \solution{} and an occurrence of $X$ in the equation.
Then there is at most one explicit word whose letters are arranged against letters from this fixed occurrence of \sol X.
\end{lemma}
\begin{proof}
Without loss of generality assume that $X$ occurs within $\mathcal A_\ell$ in an equation $\mathcal A_\ell = \mathcal B_\ell$.
Suppose that $B_j$ and $B_{j'}$ (for $j' > j$) have their letters arranged against a letter from this fixed occurrence of \sol X,
see Fig~\ref{fig:protected}.
But $B_j$ and $B_{j'}$ are separated by at least one $X$ in the equation,
and whole this $X$ is also arranged against this fixed occurrence of $X$, contradiction.
\qedhere
\end{proof}

As a second remark, observe that tests include not only explicit letters from \sol {\mathcal A_\ell} and \sol {\mathcal B_\ell}
but also letters from \sol X.
In the following we will focus on tests in which at least one letter comes from an explicit word.
It is easy to show that the time spent on other tests is at most as large as time spent on those tests.
This follows from the fact that such other tests boil down to comparison of long blocks of $a$
and the previous test is of a different type, so we can account the comparison between two long blocks of $a$ to the previous test.
However, our fast testing procedures in some times makes a series of tests in $\Ocomp(1)$ time, so this argument can be made precise
only after the explanation of the details of various testing optimisations.
For this reason the proof of Lemma~\ref{lem: only letter tests} is delayed till the end of this section.

\begin{lemma}
\label{lem: only letter tests}
Suppose that we can perform all tests in which at least one letter comes from an explicit word in $\Ocomp(n)$ time.
Then we can perform all test in $\Ocomp(n)$ time.
\end{lemma}
Thus, in the following section we consider only the tests in which at least one letter comes from an explicit word.

\subsubsection{Test types}
Suppose that for a substitution \solution{} a letter from $A_i$ is tested against a letter from \sol{XB_j}
or a letter from $B_j$ is tested against a letter from \sol{XA_i}
(the special case, when there is no explicit word after $X$ is explained later).
We say that this test is:

\begin{description}
	\item[protected] if at least one of $A_i$, $A_{i+1}$, $B_j$, $B_{j+1}$ is long;
	\item[failed] if $A_i$, $A_{i+1}$, $B_j$ and $B_{j+1}$ are short
	and a mismatch for \solution{} is found till the end of $A_{i+1}$ or $B_{j+1}$;
	\item[aligned] if $A_i = B_j$ and $A_{i+1} = B_{j+1}$, all of them are short and the first letter of $A_i$ is tested against the first letter of $B_j$;
	\item[misaligned] if all of $A_i$, $A_{i+1}$, $B_j$, $B_{j+1}$ are short,
	$A_{i+1} \neq A_i$ or $B_{j+1} \neq B_j$ and this is not an aligned nor failed test;
	\item[periodical] if $A_{i+1} = A_{i}$, $B_{j+1} = B_j$, all of them are short and this is not an aligned nor failed test.
\end{description}
This classification does not apply to the case, when a letter from $A_i$ is tested against letter from $X$ that is not followed by an explicit word.
There are two cases:
\begin{itemize}
	\item If $A_i$ is not followed by $X$ in the equation then $A_i$ is a last word,
	in particular it is long.
	Therefore this test is protected.
	\item If $A_i$ is followed by $X$ then there is a mismatch till the end of $A_iX$, so this test is failed.
\end{itemize}

Observe that `failed test' does not mean a mismatch, just a fact that soon there will be a mismatch.
The protected, misaligned and failed tests are done in a letter-by-letter way,
while the aligned and periodical tests are made in larger groups
(in $\Ocomp(1)$ time per group, this of course means that we use some additional data structures).

It is easy to show that there are no other tests, see Lemma~\ref{lem:no other test}.
We separately calculate the cost of each type of tests.
As some tests are done in groups, we distinguish between number of tests of a particular type (which is the number of letter-to-letter comparisons)
and the time spent on test of a particular type (which may be smaller, as group of tests are performed in $\Ocomp(1)$ time);
the latter includes also the time needed to create and sustain the appropriate data structures.

For failed tests note that they take constant time per phase and we know that there are $\Ocomp(\log n)$ phases.
For protected tests, we charge the cost of the protected test to the long word
and only $\Ocomp(|C|)$ such tests can be charged to one long word $C$ in a phase.
On the other hand, each long word is shortened by a constant factor in a phase, see Lemma~\ref{lem:reducing length},
and so this cost can be charged to those removed letters and thus the total cost
of those tests (over the whole run of \algsonevar) is $\Ocomp(n)$.

In case of the misaligned tests, it can be shown that \solution{} in this case is small
and that it is tested at the latest $\Ocomp(1)$ phases after the last of $A_{i+1}$, $A_i$, $B_{i+1}$, $B_i$ becomes short,
so this cost can be charged to, say, $B_i$ becoming short and only $\Ocomp(1)$ such tests are charged to this $B_i$
(over the whole run of the algorithm). Hence the total time of such tests is $\Ocomp(n)$.

For the aligned tests, consider the consecutive aligned tests, they correspond to comparison of
$A_iXA_{i+1}\dots A_{i+k}X$ and $B_jXB_{j+1}\dots B_{j+k}X$, where $A_{i + \ell} = B_{j + \ell}$ for $\ell = 1, \ldots , k$.
So to perform them efficiently, 
it is enough to identify the maximal (syntactically) equal
substrings of the equation and from Lemma~\ref{lem:words are equal}
it follows that this corresponds to the (syntactical) equality of substrings in the original equation.
Such an equality can be tested in $\Ocomp(1)$ using a suffix array constructed for the input equation
(and general lcp queries on it).
To bound the total running time it is enough to notice that the previous test is either misaligned or protected.
There are $\Ocomp(n)$ such tests in total,  so the time spent on aligned tests is also linear.

For the periodical test suppose that we are to test the equality of (suffix of) $\sol{(A_iX)^\ell}$
and (prefix of) $\sol{X(B_jX)^k}$.
If $|A_i| = |B_j|$ then the test for $A_{i+1}$ and $B_{j+1}$ is the same as for $A_i$ and $B_j$ and so can be skipped.
If $|A_i| > |B_j|$ then the common part of $\sol{(A_iX)^\ell}$ and $\sol{X(B_jX)^k}$ have periods
$|\sol{A_iX}|$ and $|\sol{B_jX}|$ and consequently has a period $|A_i| - |B_j| \leq \bound$.
So it is enough to test first common $|A_i| - |B_j|$ letters and check whether
$|\sol{A_iX}|$ and $|\sol{B_jX}|$ have period $|A_i| - |B_j|$,
which can be checked in $\Ocomp(1)$ time.

This yields that the total time of testing is linear. The details are given in the next subsections.

We begin with showing that indeed each test is either failed, protected, misaligned or periodical.

\begin{lemma}
\label{lem:no other test}
Each test is either failed, protected, misaligned, aligned or periodical.
Additionally, whenever a test in made, in $\Ocomp(1)$ time we can establish,
what type of test this is.
\end{lemma}
\begin{proof}
Without loss of generality, consider a test of a letter from $A_i$ and from $\sol{XB_j}$.
If any of $A_{i+1}$, $B_{j+1}$, $A_i$ or $B_j$ is long then it is protected
(this includes the case in which some of $A_{i+1}$, $B_j$, $B_{j+1}$ does not exist).
Concerning the running time, for each explicit word we keep a flag, whether it is short or long.
Furthermore, as each explicit word has a link to its successor and predecessor,
we can establish whether any of $A_{i+1}$, $B_{j+1}$, $A_i$ or $B_j$ is long in $\Ocomp(1)$ time.

So consider the case in which all $A_{i+1}$, $B_{j+1}$, $A_i$ or $B_j$ (if they exist) are short,
which also can be established in $\Ocomp(1)$ time.
It might be that this test is failed (again, some of the words $A_{i+1}$, $B_j$, $B_{j+1}$ may not exist),
too see this we need to make some look-ahead tests, but this can be done in $\Ocomp(\bound)$ time
(we do not treat those look-aheads as tests, so there is not recursion here).

Otherwise, if the first letter of $A_i$ and $B_j$ are tested against each other and $A_i = B_j$ and $A_{i+1} = B_{j+1}$
then the test is aligned (clearly this can be established in $\Ocomp(1)$ time using look-aheads).
Otherwise, if $A_{i+1} \neq A_i$ or $B_{j+1} \neq B_j$ then it is misaligned (again, $\Ocomp(1)$ time for look-aheads).
In the remaining case $A_{i+1} = A_i$ and $B_{j+1} = B_j$, so this is a periodical test.
\qedhere
\end{proof}

\subsubsection{Failed tests}
We show that in total there are $\Ocomp(\log n)$ failed tests.
This follows from the fact that there are $\Ocomp(1)$ substitutions tested per phase and there are $\Ocomp(\log n)$ phases.
\begin{lemma}
\label{lem:failed is linear}
The number of all failed tests is $\Ocomp(\log n)$ over the whole run of \algsonevar.
\end{lemma}
\begin{proof}
As noticed, there are $\Ocomp(1)$ substitutions tested per phase.
Suppose that the mismatch is for the letter from $A_i$ and a letter from $XB_j$
(the case of $XA_i$ and $B_j$ is symmetrical).
Then the failed tests include at least one letter from $XA_{i-1}XA_i$ or $XB_{j-1}XB_jX$,
assuming they come from a short word.
There are at most $4 \bound$ failed tests that include a letter from $A_{i-1}$, $A_i$, $B_{j-1}$, $B_{j}$
(as the test is failed then in particular this explicit word is short).
Concerning the tests including the short occurrences of $X$ in-between them,
observe that by Lemma~\ref{lem: against X} each such $X$ can have tests with at most one short word,
so this gives additional $5 \bound$ tests.
Since $\bound = \Ocomp(1)$, we conclude that there are $\Ocomp(1)$ failed tests per phase and so $\Ocomp(\log n)$ failed tests in total,
as there are $\Ocomp(\log n)$ phases, see Lemma~\ref{lem:reducing length}.
\qedhere
\end{proof}

\subsubsection{Protected tests}
As already claimed, the total number of protected tests is linear in terms of length of long words:
to show this it is enough to charge the cost of the protected test to the appropriate long word and
see that a long word $A$ can be charged only $|A|$ such tests for test including letters from $A$
and $\Ocomp(1)$ letters from neighbouring short words, which yields $\Ocomp(|A|)$ tests.
As the length of the long words drops by a constant factor, summing this up over all phases in which this
explicit word is long yields $\Ocomp(n)$ tests in total.

\begin{lemma}
\label{lem:protected is linear}
In one phase the total number of protected tests is proportional to the length
of the long words.
In particular, there are $\Ocomp(n)$ such test during the whole run of \algsonevar.
\end{lemma}
\begin{proof}
As observed in Lemma~\ref{lem: only letter tests} we can consider only tests in which at least one letter comes from an explicit word.
Suppose that a letter from $A_i$ takes part in the protected test (the argument for a letter from $B_j$ is similar, it is given later on)
and it is tested against a letter from $XB_j$,
then one of $A_i$, $A_{i+1}$, $B_j$, $B_{j+1}$ is long, we charge the cost according to this order,
i.e.\ we charge it to $A_i$ if it is long, if $A_i$ is not but $A_{i+1}$ is long, we charge it to $A_{i+1}$,
if not then to $B_{j}$ if it is long and otherwise to $B_{j+1}$.
The analysis and charging for a test of a letter from $B_j$ is done in a symmetrical way
(note that when the test includes two explicit letters, we charge it twice, but this is not a problem).

Now, fix some long word $A_i$, we estimate, how many protected tests can be charged to it.
It can be charged with cost of tests that include its own letters,
so $|A_{i}|$ tests.
When $A_{i-1}$ is short, it can also charge tests in which its letters take part.
As it is short, it is at most $\Ocomp(\bound) = \Ocomp(1)$ such tests.

Also some $\mathcal B$ words can charge the cost of tests to $A_i$, we can count only the test in which letters
from $A_{i}$ do not take part.
This can happen in two situations:
letters tested against $XA_i$ and letters tested against $XA_{i-1}$ (in which case we additionally assume that $A_{i-1}$ is short).
We have already accounted the tests made against $A_{i-1}$ and $A_i$ and by Lemma~\ref{lem: against X}
for each occurrence of $X$ there is at most one explicit short word whose letters are tested against this occurrence of $X$.
So there are additionally at most $2 \bound$ tests of this form.

So in total $A_{i}$ is charged only $\Ocomp(|A_{i}|)$ in a phase.
From Lemma~\ref{lem:reducing length} the sum of lengths of long words drops by a constant factor in each phase,
and as in the input it is at most $n$, the total sum of number of protected tests is $\Ocomp(n)$.
\qedhere
\end{proof}

\subsubsection{Misaligned tests}
On the high level, in this section we want to show that if there is a misaligned test then the tested solution is small
and use this fact for accounting the cost of such tests.
However, this statement is trivial, as we test only solutions of the form $a^k$ for some $k$, which are always small.
To make this statement more meaningful, we generalise the notion of a misaligned test for arbitrary substitutions,
not only the tested one.
In this way two explicit words $A_i$ and $B_j$ can be misaligned for a substitution \solution.
We show three properties of this notion:
\begin{enumerate}[M1]
	\item \label{M1} If there is a misaligned test for a substitution \solution{}
	for a letter from $A_i$ against letter in $XB_j$ or a letter from $B_j$ against letter from
	$XA_i$ then $A_i$ and $B_j$ are misaligned for \solution.
	This is shown in Lemma~\ref{lem: definition reformulation}.
	\item \label{M2} If there are misaligned words $A_i$ and $B_j$ for a solution \solution{} then \solution{} is small,
	as shown in Lemma~\ref{lem:low height for misaligned}.
	\item \label{M3} If $A_i$ and $B_j$ are misaligned for \solution{} in a phase $\ell$
	then \solution{} is reported in phase $\ell$ or the corresponding words $A_i'$ and $B_j'$ in phase $\ell+1$
	are also misaligned for the corresponding $\solution'$, see Lemma~\ref{lem:misaligned earlier}.
\end{enumerate}

Those properties are enough to improve the testing procedure so that one $\mathcal A$ $i$-word (or $\mathcal B$ $j$-word)
takes part in only $\Ocomp(1)$ misaligned tests:
suppose that $A_i$ becomes small in phase $\ell$. Then all solutions, for which it is misaligned with some $B_j$,
are small by~\Mref{1}.
Hence, by Theorem~\ref{thm:solution with small height},
all of those solutions are reported (in particular: tested) within the next $c$ phases, for some constant $c$.
Thus, if $A_i$ takes part in a misaligned test (for \solution) in phase $\ell' > \ell + c$ then \solution{} is not a solution:
by~\Mref{3} also in phase $\ell$ the $A_i$ and $B_j$ were misaligned (for the corresponding solution $\solution'$),
and solution $\solution'$ was reported before phase $\ell'$.
Hence we can immediately terminate the test;
therefore $A_i$ can take part in misaligned tests in phases $\ell$, $\ell+1$, \ldots, $\ell + c$, i.e.\ $\Ocomp(1)$ ones.
This plan is elaborated in this section, in particular, some technical details (omitted in the above description) are given.

We say that $A_i$ and $B_j$ that are blocks from two sides of one equations $\mathcal A_\ell = \mathcal B_\ell $
are \emph{misaligned for a substitution} \solution{} if
\begin{itemize}
	\item a mismatch for \solution{} is not found till the end of $A_{i+1}$ or $B_{j+1}$;
	\item all $A_{i+1}$, $A_i$, $B_{j+1}$ and $B_j$ are short;
	\item either $A_i \neq A_{i+1}$ or $B_j \neq B_{j+1}$
	\item it does not hold that $A_i = B_j$ and $A_{i+1} = B_{j+1}$
	and the first letter of $A_i$ is at the same position as the first letter of $B_j$
	under substitution \solution;
	\item the position of the first letter of $A_i$ in \sol{\mathcal A _\ell }
	is among the position of \sol{XB_j} in \sol{\mathcal B _\ell} or, symmetrically, 
	the position of the first letter of $B_j$ in \sol{\mathcal B _\ell }
	is among the position of \sol{XA_i} in \sol{\mathcal A _\ell}
\end{itemize}

We show~\Mref{1}, which shows that the definitions of misaligned blocks and misaligned tests are reformulations of each other.

\begin{lemma}
\label{lem: definition reformulation}
If a letter from $A_i$ is tested (for \solution) against a letter from $XB_j$ and this test is misaligned
then $A_i$ and $B_j$ are misaligned for \solution;
similar statement holds for letters from $B_j$.
\end{lemma}
\begin{proof}
This is just a reformulation of a definition (we consider only the case of letters from $A_i$,
the argument for letters from $B_j$ is symmetrical):

\begin{itemize}
	\item Since this is not a failed test, there is no mismatch till the end of $A_{i+1}$ and $B_{j+1}$.
	\item As this is not a protected test, all $A_i$, $A_{i+1}$, $B_j$ and $B_{j+1}$ are short.
	\item As this is a misaligned test, either $A_i \neq A_{i+1}$ or $B_j \neq B_{j+1}$.
	\item As this is not an aligned test, either $A_i \neq B_j$ or $A_{i+1} \neq B_{j+1}$ or the first letter of $A_i$
	is not at the same position as the first letter of $B_j$ (both under \solution).
	\item By the choice of $B_j$, the first position of $A_i$ under \solution{} is among the positions of $XB_j$ (under \solution).
	\qedhere
\end{itemize}
\end{proof}

We move to showing~\Mref{2}.
It follows by considering \sol{XA_iXA_{i+1}X} and \sol{XB_jXB_{j+1}X}.
The large amount of \sol X in it allow to show the periodicity of fragments of \sol X and in the end, that \solution{} is small.

\begin{lemma}
\label{lem:low height for misaligned}
When the $A_i$ and $B_j$ are misaligned for a solution \solution{} then \solution{} is small.
\end{lemma}
\begin{proof}
Suppose that $A_i$ and $B_j$ are from an equation $\mathcal A_\ell = \mathcal B_\ell $.
In the proof we consider only one of the symmetric cases, in which $A_i$ is begins not later than $B_j$
(i.e.\ the first letter of $A_i$ is arranged against the letter from $XB_j$).

There are two main cases: either some of $A_i$, $A_{i+1}$, $B_j$ and $B_{j+1}$ has some of its letters
arranged against an explicit word or all those words are arranged against (some occurrences) of $X$.

\begin{figure}
	\centering
		\includegraphics{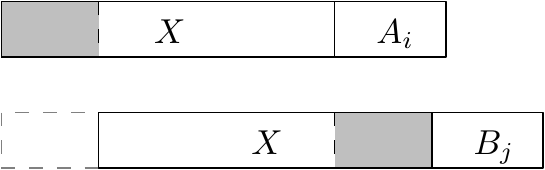}
	\caption{A letter from $B_j$ is arranged against the letter from $A_i$. The period of $\sol X$ is in grey.}
	\label{fig:misaligned_in}
\end{figure}

\subsubsection*{One of the words has some of its letters arranged against an explicit word.}
We claim that in this case \solution{} has a period of length at most $\bound$, in particular, it is small.
First of all observe that it is not possible that \emph{each} of $A_i$, $A_{i+1}$, $B_j$ and $B_{j+1}$
has \emph{all} of its letters arranged against letters of an explicit word:
since $A_i$ is arranged against $XB_j$ this would imply that $A_i$ is arranged against $B_j$
(in particular, their first letters are at corresponding positions)
and (as no mismatch is found till end of $A_i$ and $B_j$) so $A_i = B_j$.
Similarly, $A_{i+1} = B_{j+1}$.
This contradicts the assumption that $A_i$ and $B_j$ are misaligned.

Thus, there is a word among $A_i$, $A_{i+1}$, $B_j$ and $B_{j+1}$, say $B_j$,
that is partially arranged against an explicit word and partially against $X$
(note that this explicit word does not have to be among $A_i$, $A_{i+1}$, $B_j$ and $B_{j+1}$),
see Figure~\ref{fig:misaligned_in}.
As each explicit words is proceeded and succeeded by $X$, it follows that \sol X has a period at most $\bound$.

\begin{figure}
	\centering
		\includegraphics{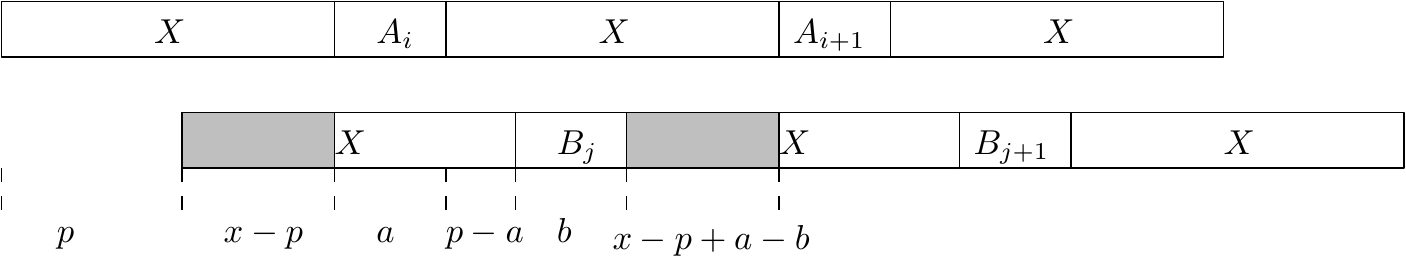}
	\caption{The letters of $A_i$, $A_{i+1}$, $B_j$ and $B_{j+1}$ are arranged against the letters from \sol X.
	The lengths of fragments of text are beneath the figure, between dashed lines.
	Comparing the positions of the first and second \sol X yields that $p$ is a period of \sol X,
	second and third that $x-p+a$ while the third and fourth that $p - a + b$ is.
	The borders of \sol X corresponding to the first and third one are marked in grey.}
	\label{fig:misaligned}
\end{figure}

\subsubsection*{All words have all their letters arranged against occurrences of $X$.}
In the following we assume that letters from $A_i$, $A_{i+1}$, $B_{j}$ and $B_{j+1}$ are arranged against the letters from \sol X.
Observe that due to Lemma~\ref{lem: against X} this means that whole $A_i$ is arranged against \sol X preceding $B_j$,
the $B_j$ against \sol X preceding $A_{i+1}$, whole $A_{i+1}$ against \sol X preceding $B_{j+1}$ and 
whole $B_{j+1}$ against \sol X succeeding $A_{i+1}$.

Let $a = |A_{i}|$, $b = |B_{j}|$  and $x = |\sol X|$, as in Fig.~\ref{fig:misaligned}.
There are three cases: $a > b$, $a < b$ and $a = b$, we consider them separately.

Consider first the case in which $a > b$, see Fig.~\ref{fig:misaligned}.
Let $p$ denote the offset between the \sol X preceding $A_{i}$ and the one proceeding $B_{j}$;
then $\sol X$ has a period $p$.
Similarly, when we consider the \sol X succeeding $A_i$ and the one succeeding $B_j$ we obtain that
the offset between them is $p - a + b$, which is also a period of \sol X.
Those offsets correspond to borders (of \sol X) of lengths $x - p$ and $x - p + a - b$, see Fig.~\ref{fig:misaligned}.
Then the shorter border (of length $x - p$) is also a border of the longer one (of length $x-p+a-b$),
hence the border of length $x - p + a - b$ has a period $a - b$,
so it is of the form $w^ku$, where $|w| = a - b$ and $|u| < a - b$.
Now, the prefix of \sol X of length $x-p + a$ is of the form $w^ku'$, for some $u'$ of length less than $a$
(as this is a prefix of length $x-p+a-b$ extended by the following $b$ letters).
When we compare the positions of \sol X preceding $B_j$ and the one succeeding $A_i$ we obtain that $\sol X$ has a period $x - p + a$
so the whole \sol X is of the form $(w^ku')^{\ell}w'$, where $w'$ is a prefix of $w^ku'$, hence \solution{} is small:
$w$ and $u$ are of length at most \bound, as $w'$ is a prefix of $w^ku$, either it is a prefix of $w^k$, so it is of the form $w^{k'} w''$
where $w''$ is a prefix of $w$, or it includes the whole $w^k$, so it is of the form $w^ku''$, where $u''$ is a prefix of $u$.

Consider the symmetric case, in which $b > a$ and again use Fig.~\ref{fig:misaligned}.
The same argument as before shows that $p$ and $p - a + b$ are periods of \sol X
and the corresponding borders are of length $x - p$ and $x - p + a - b$.
Now, the shorter of them (of length $x - p + a - b$) is a border of longer of them (of length $x - p$),
so the prefix of length $x - p$ of \sol X has a period $b - a$, so it is of the form $w^ku$,
where $|w| = b-a$ and $|u| < b-a$.
Hence the prefix of length $x - p + a$ is of the form $w^ku'$ for some $u'$ of length less than $b$.
As in the previous case, \sol X has a period $x - p + a$ and so the whole \sol X is of the form $(w^ku')^{\ell}w'$,
where $w'$ is a prefix of $w^ku'$, hence \solution{} is small.

Consider now the last case, in which $|A_{i}| = |B_{j}|$.
If $|A_{i+1}| \neq |A_{i}|$ then $|B_{j}| \neq |A_{i+1}|$ and we can repeat the same argument as above,
with $B_{j}$ and $A_{i+1}$ taking the roles of $A_{i}$ and $B_{j}$, which shows that \solution{} is small.
So consider the case in which $|A_{i+1}| = |A_{i}|$.
If $|B_j| \neq |B_{j+1}|$ then again, repeating the argument as above for $A_{i+1}$ and $B_{j+1}$ yields that \solution{} is small.
So we are left with the case in which $|A_{i+1}| = |A_i| = |B_{j}| = |B_{j+1}|$.
Then $A_{i+1}$ is arranged against the same letters in \sol X as $A_{i}$
and $B_{j+1}$ is arranged against the same letters in \sol X as $B_{j}$.
As there is no mismatch till the end of $A_{i+1}$ and $B_{j+1}$,
we conclude that $A_{i+1} = A_i$ and $B_{j+1} = B_j$ contradicting the assumption that $A_i$ and $B_j$ are misaligned,
so this case is non-existing.
\qedhere
\end{proof}

We now show that if $A_i$ and $B_j$ are misaligned for \solution{}
then they were (for a corresponding solution) in the previous phase (assuming that all involved words were short).
This is an easy consequence of the way explicit words are modified (we prepend and append the same letters
and compress all explicit words in the same way).

\begin{lemma}
\label{lem:misaligned earlier}
Suppose that $A_i$ and $B_j$ are misaligned for a solution \solution.
If at the previous phase all $A_{i+1}'$, $A_i'$, $B_{j+1}'$ and $B_j'$ were short
then $A_i'$ and $B_j'$ were misaligned for the corresponding solution $\solution'$.
\end{lemma}
\begin{proof}
We verify the conditions on misaligned words point by point:
\begin{itemize}
	\item Since $\solution'$ is a solution, there is no mismatch.
	\item By the assumption, all  $A_{i+1}'$, $A_i'$, $B_{j+1}'$ and $B_j'$ are short.
	\item We know that either $A_i \neq A_{i+1}$ or $B_j \neq B_{j+1}$ and so by Lemma~\ref{lem:words are equal}
	either $A_i' \neq A_{i+1}'$ or $B_j' \neq B_{j+1}'$ (observe that none of them is the last nor first, as they are all short).
	\item Suppose that $A_i' = B_j'$, $A_{i+1}' = B_{j+1}'$ and under $\solution'$ the first letters of $A_i'$ and $B_j'$
	are arranged against each other.
	By Lemma~\ref{lem:words are equal} it follows that $A_i = B_j$, $A_{i+1} = B_{j+1}$.
	Observe that left-popping and right popping preserves the fact that the first letters of
	$\mathcal A$ $i$-word and $\mathcal B$ $j$-word are arranged against each other for $\solution'$
	(as \sol {\mathcal A} and $\solution'(\mathcal A')$ are the same words).(as \sol {\mathcal A} and $\solution'(\mathcal A')$ are the same words)
	As $\solution'$ is a solution, the same applies to pair compression and block compression.
	Hence, the first letters of $A_i$ and $B_j$ are arranged against each other, contradiction
	with the assumption that $A_i$ and $B_j$ are misaligned.	
	\item
	Suppose that the first letter of $A_i$ is arranged against a letter from \sol{XB_j}.
	Consider, how $A_i'$ and $XB_j'$ under $\solution'$ are transformed to $A_i$ and $X B_j$ under \solution.
	As in the above item, popping letters does not influence whether the first letter of $\mathcal A$ $i$-word is arranged against
	letter from \sol X and $\mathcal B$ $j$-word
	(as \sol {\mathcal A} and $\solution'(\mathcal A')$ are the same words).
	Since $\solution'$ is a solution, the same applies also to pair and block compression.
	So the position of the first letter of $A_i$ is among the position of \sol{XB_j}
	if and only if the first letter of $A_i'$ is arranged against a letter from $\solution'(XB_j')$.

	The case in which the position of the first letter of $B_j$	is among the position of \sol{XA_i} is shown in a symmetrical way.
	\qedhere
\end{itemize}
\qedhere
\end{proof}

Now we are ready to give the improved procedure for testing and estimate the number of the misaligned tests in it.

\begin{lemma}
\label{lem:misaligned is linear}
There are $\Ocomp(n)$ misaligned tests during the whole run of \algsonevar.
\end{lemma}
\begin{proof}
Consider a tested solution $\solution$ and a misaligned test for a letter from $A_i$ against a letter from $XB_j$
(the case of test of letters from $B_j$ tested against $XA_i$ the argument is the same).
Let $\ell$ be the number of the first phase, in which all $\mathcal A$ $i$-word, $\mathcal A$ $i+1$-word,
$\mathcal B$ $j$-word and $\mathcal B$ $j+1$-word are short.
We claim that this misaligned test happens between $\ell$-th and $\ell+c$ phase,
where $c$ is the $\Ocomp(1)$ constant from Theorem~\ref{thm:solution with small height}.

Let $A_i'$ and $B_j'$ be the corresponding words in the phase $\ell$.
Using induction on Lemma~\ref{lem:misaligned earlier} it follows that $A_i'$ and $B_j'$ are misaligned for $\solution'$.
Thus by Lemma~\ref{lem:low height for misaligned} the $\solution'$ is small and thus by Theorem~\ref{thm:solution with small height}\
it is reported till phase $\ell+c$. So it can be tested only between phases $\ell$ and $\ell+c$, as claimed.

This allows an improvement to the testing algorithm:
whenever (say in phase $\ell$) a letter from $A_i$ has a misaligned test against a letter from $\sol{X B_j}$
we can check (in $\Ocomp(1)$ time), in which turn $\ell'$ the last among $\mathcal A$ $i$-word, $\mathcal A$ ${i+1}-word$, $\mathcal B$ $j$-word
and $\mathcal B$ ${j+1}$ word became small
(it is enough to store for each explicit word the number of phase in which it became small).
If $\ell' + c <\ell$ then we can terminate the test, as we know already that \solution{} is not a solution.
Otherwise, we continue.

Concerning the estimation of the cost of the misaligned tests (in the setting as above), there are two cases:
\begin{description}
	\item[The misaligned tests that lead to the rejection of \solution]
	This can happen once per tested solution	and there are $\Ocomp(\log n)$ tested solution in total
	($\Ocomp(1)$ per phase and there are $\Ocomp(\log n)$ phases).
	\item[Other misaligned tests]
	The cost of the test (of a letter from $A_i$ tested against $\sol{XB_j}$)
	is charged to the last one among $\mathcal A$ $i$-word, $\mathcal A$ ${i+1}$-word, $\mathcal B$ $j$-word
	and $\mathcal B$ ${j+1}$-word that became short.
	By the argument above, this means that this word became short within the last $c$ phases.
	
	Let us calculate, for a fixed $\mathcal A$ $i$ word (the argument for $\mathcal B$ $j$-word is symmetrical)
	how many aligned tests of this kind can be charged to this word.
	They can be charged only within $c$ phases after this word become short.
	In a fixed phase we test only a constant (i.e.\ $5$) substitutions.
	For a fixed substitution, $A_i$ can be charged the cost of tests in which letters from $A_i$ or $A_{i-1}$ are involved
	(providing that $A_i$/$A_{i-1}$ is short), which is at most $2 \bound$.
	They can be charged also the tests from letters from $B_j$ that is aligned against $X$ proceeding $A_{i-1}$
	or $X$ proceeding $A_i$ (providing that $B_j$ as well as $A_{i-1}$ are short).
	Note that there is only one $B_j$ whose letter are aligned against $X$ proceeding $A_{i-1}$
	and one for $X$ proceeding $A_i$, see Lemma~\ref{lem: against X},
	so when they are short this gives additional $2 \bound$ tests.

	This yields that one $\mathcal A$ $i$ word is charged $\Ocomp(\bound) = \Ocomp(1)$ tests in total.
	Summing over all words in the instance yields the claim of the lemma.
\qedhere
\end{description}
\end{proof}

\subsubsection{Aligned tests}
Suppose that we make an aligned test, without loss of generality consider the first such test in a sequence of aligned tests.
Let it be between the first letter of $A_i$ and the first letter in $B_j$
(both of those words are short).
For this $A_i$ and $B_j$ we want to perform the whole sequence of successive aligned tests at once,
which corresponds of jumping to $A_{i+k}$ and $B_{j+k}$ within the same equation such that
\begin{itemize}
	\item $A_{i +\ell} = B_{i + \ell}$ for $0 \leq \ell < k$;
	\item $A_{i+k} \neq B_{j+k}$ or one of them is long or $A_{i+k}X$  or $B_{j+k}X$ ends one side of the equation.
\end{itemize}
Note that this corresponds to a syntactical equality of fragments of the equation,
which, by Lemma~\ref{lem:words are equal} is equivalent to a syntactical equality of fragments of the original equation.
We preprocess (in $\Ocomp(n)$ time) the input equation
(building a suffix array equipped with a structure answering general lcp queries)
so that in $\Ocomp(1)$ we can return such $k$ as well as the links to $A_{i+k}$ and $B_{j+k}$.
In this way we perform all equality tests for $A_iXA_{i+1}X\ldots A_{i+k-1}X = B_{j}XB_{j+1}X\ldots XB_{j+k-1}X$
in $\Ocomp(1)$ time.

To simplify the considerations, when $A_iX$ ($B_jX$) ends one side of the equation,
we say that this $A_i$ ($B_j$, respectively) is \emph{almost last} word.
Observe that in a given equation exactly one side has a last word and one an almost last word.

\begin{lemma}
\label{lem: data structure aligned test}
In $\Ocomp(n)$ we can build a data structure which given equal $A_i$ and $B_j$ in $\Ocomp(1)$ time
returns the smallest $k \geq 1$ and links to $A_{i+k}$ and $B_{j+k}$
such that $A_{i+k} \neq B_{j+k}$ or one of $A_{i+k}$, $B_{j+k}$ is a last word or one of $A_{i+k}$, $B_{j+k}$ is an almost last word.
\end{lemma}
Note that it might be that some of the equal words $A_{i + \ell} = B_{i + \ell}$ are long,
and so their tests should be protected (also, the tests for some neighbouring words).
So in this way we also make some free protected tests, but this is not a problem.
Furthermore, the returned $A_{i+k}$ and $B_{j+k}$ are guaranteed to be in the same equation.
\begin{proof}
First of all observe that for $A_i$ and $B_j$ it is easy to find the last word in their equation as well as the almost last word of the equation:
when we begin to read a particular equation, we have the link to both the last word and the almost last word of this equation
and we can keep them for the testing of this equation.
We also know the numbers of those words so we can also calculate the respective candidate for $k$.
So it is left to calculate the minimal $k$ such that $A_{i + k} \neq A_{j + k}$.

Let $A_{i}'$, $B_{j}'$ etc.\ denote the corresponding original words of the input equation.
Observe that by Lemma~\ref{lem:words are equal} it holds that $A_{i + \ell}' = B_{j + \ell}'$ if and only if $A_{i + \ell} = B_{j + \ell}$
as long as none of them is last or first word.
Hence, it is enough to be able to answer such queries for the input equation: if the returned word is in another equation
then we should return the last or almost last word instead.

To this end we build a suffix array~\cite{suffixarrays} for the input equation,
i.e.\ for $A_1'XA_2'X\ldots A_{n_{\mathcal A}}'XB_1'XB_2'X\ldots B_{n_{\mathcal B}}'\$$.
Now, the lcp query for suffixes $A_i'\ldots \$$ and $B_j'\ldots \$$\
returns the length of the longest common prefix.
We want to know what is the number of explicit words in the common prefix,
which corresponds to the number of $X$s in this common prefix.
This information can be easily preprocessed and stored in the suffix array:
for each position $\ell$ in $A_1'XA_2'X\ldots A_{n_{\mathcal A}}'XB_1'XB_2'X\ldots B_{n_{\mathcal B}}'\$$
we store, how many $X$s are before it in the string and store this in the table $\textsl{prefX}$.
Then when for a suffixes beginning at positions $p$ and $p'$ we get that their common prefix is of length $\ell$,
the $\textsl{prefX}[p+\ell] - \textsl{prefX}[p]$ is the number of $X$s in the common prefix in such a case.
If none of $A_i$, $A_{i+1}$, \ldots, $A_{i+k}$ nor $B_j$, $B_{j+1}$, \ldots, $B_{j+k}$ is the last word
nor it ends the equation (i.e.\ they are all still in one equation)
by Lemma~\ref{lem:words are equal} the $k$ is the answer to our query
(as $A_i = B_j$, $A_{i+1} = B_{j+1}$,\ldots and $A_{i+k} \neq B_{j+k}$ and none of them is a last word, nor none of them ends the equation).
To get the actual links to those words, at the beginning of the computation we make a table, which for each $i$
return the pointer to $\mathcal A$ $i$-word and $\mathcal B$ $i$-word has the link to this word.
As we know $i$, $j$ and $k$ we can obtain the appropriate links in $\Ocomp(1)$ time.
So it is left to compare the value of $k$ with the value calculated for the last word and almost last word
and choose the one with smaller $k$ and the corresponding pointers.
\qedhere
\end{proof}

Using this data structure we perform the aligned tests is in the following way:
whenever we make an aligned test (for the first letter of $A_i$ and the first letter of $B_j$),
we use this structure, obtain $k$ and jump to the test of the first letter of $A_{i+k}$ with the first letter of $B_{j+k}$
and we proceed with testing from this place on.
Concerning the cost, by easy case analysis it can be shown that the test right before the first of sequence of aligned tests
(so the test for the last letters of $A_{i-1}$ and $B_{j-1}$) is either protected or misaligned.
There are only $\Ocomp(n)$ such tests (over the whole run of \algsonevar), so the time spend on aligned tests is $\Ocomp(n)$ as well.

\begin{lemma}
\label{lem:aligned is linear}
The total cost aligned test as well as the usage of the needed data structure is $\Ocomp(n)$.
\end{lemma}
\begin{proof}
We formalise the discussion above.
In $\Ocomp(1)$ we get to know that this is an aligned test, see Lemma~\ref{lem:no other test}.
Then in $\Ocomp(1)$, see Lemma~\ref{lem: data structure aligned test},
we get the smallest $k$ such that $A_{i+k} \neq B_{j+k}$ or one of them is an almost last word for this equation or the last word for this equation.
We then jump straight to the test for the first letter of $A_{i+k}$ and $B_{j+k}$.

Consider $A_{i-1}$ and $B_{j-1}$ we show that the test for their last letters (so the test immediately before the first aligned one)
is protected or misaligned.
By Lemma~\ref{lem:no other test} it is enough to show that it is not aligned, nor periodic, nor failed.
\begin{itemize}
	\item If it were failed then also the test for the first letters of $A_{i}$ and $B_j$ would be failed.
	\item It cannot be aligned, as we chose $A_i$ and $B_j$ as the first in a series of aligned tests.
	\item If it were periodic, then $A_{i-1} = A_i$ and $B_{j-1} = B_j$ while by assumption $A_i = B_j$,
	which implies that this test is in fact aligned, which was already excluded.
\end{itemize}
Hence we can associate the $\Ocomp(1)$ cost of whole sequence of aligned test to the previous test,
which is misaligned or protected.
Clearly, one misaligned or protected test can be charged with only one sequence of aligned tests (as it is the immediate previous test).
By Lemma~\ref{lem:protected is linear} and \ref{lem:misaligned is linear} in total there are $\Ocomp(n)$ misaligned and protected tests.
Thus in total all misaligned tests take $\Ocomp(n)$ time.
\qedhere
\end{proof}

\subsubsection{Periodical tests}
The general approach in case of periodical tests is similar as for the aligned tests: we would like to
perform all consecutive periodical tests in $\Ocomp(\bound)$ time and show that the test right before this sequence of periodic tests is either
protected or misaligned.
As in case of aligned tests, the crucial part is the identification of a sequence of consecutive periodical tests.
To identify them quickly, we keep for each short $A_i$ the value $k$ such that $A_{i+k}$ is the first word
that is different from $A_i$ or is the last word or the almost last word
(in the sense as in the previous section: $A_{i+k}$ is almost last if $A_{i+k}X$ ends the side of the equation),
as well as the link to this $A_{i+k}$.
Those are easy to calculate at the beginning of each phase.
Now when we perform a periodical test for a letter from $A_i$,
we test letters from $\sol{(AX)^k}$ against the letters from (suffix of) $\sol{X(BX)^\ell}$.
If $|A| = |B|$ then both strings are periodic with period $|\sol{AX}|$ and their equality can be tested in $\Ocomp(|A|)$.
If $|A| \neq |B|$ then we retrieve the values $k_{\mathcal A}$ and $k_{\mathcal B}$ which tell us what is repetition of $AX$ and $BX$.
If one of them is smaller than $3$ we make the test naively, in time $\Ocomp(|A| + |B|)$.
If not, we exploit the fact that $\sol{BX}^\ell$ has a period $|\sol{BX}|$
while \sol {(AX)^k} has a period $|\sol {AX}|$ and so their common fragment (if they are indeed equal)
has a period $||\sol {AX}| - |\sol {BX}|| = ||A| - |B||$.
Hence we check, whether \sol{AX} and \sol{BX} have this period and check the common fragment of this length,
which can be done in $\Ocomp(|A| + |B|)$ time.
The converse implication holds as well: if $\sol {AX}$ and $\sol {BX}$ have period $||A| - |B||$ and the first $||A| - |B||$
tests are successful then all of them are.
Concerning the overall running time, as in the case of aligned test, the test right before the first periodic test
is either protected or misaligned, so as in the previous section it can be shown
that the time spent on periodical tests is $\Ocomp(n)$ during the whole \algsonevar.

\begin{lemma}
\label{lem:periodical is linear}
Performing all periodical tests and the required preprocessing takes in total $\Ocomp(n)$ time.
\end{lemma}
\begin{proof}
Similarly as in the case of aligned tests, see Lemma~\ref{lem: data structure aligned test},
we can easily keep the value $k$ and the link to $A_{i+k}$ such that $A_{i+k}$ is the last or almost last word in this equation,
the same applies for $B_{j+k}$.
Hence it is left to show how to calculate for each short $A_i$ (and $B_j$) the $k$ such that $A_{i+k}$ is the first word
that is different from $A_i$.

At the end of the phase we list all words $A_i$ that become short in this phase, ordered from the left to the right
(this is done anyway, when we identify the new short words).
Note that this takes at most the time proportional to the length of all long words from the beginning
of the phase, so $\Ocomp(n)$ in total.
Consider any $A_{i}$ on this list (the argument for $B_j$ is identical), note that
\begin{itemize}
	\item if $A_{i+1} \neq A_i$ then $A_i$ should store $k = 1$ and a pointer to this $A_{i+1}$;
	\item if $A_{i} = A_{i+1}$ then $A_{i+1}$ also became short in this phase and so it is on the list
	and consequently $A_i$ should store $1$ more than $A_{i+1}$ and the same pointer as $A_{i+1}$.
\end{itemize}
So we read the list from the right to the left, let $A_i$ be an element on this list.
Using the above condition, we can establish in constant time the value and pointer stored by $A_i$.
This operation is performed once per block, so in total takes $\Ocomp(n)$ time.

Consider a periodic test, without loss of generality suppose that a letter from $A_i$ is tested against a letter from $XB_j$
(in particular, $A_i$ begins earlier than $B_j$),
let the $k_A$ and $k_B$ be stored by $A_i$ and $B_j$;
as this is a periodical test, both $k_A$ and $k_B$ are greater than $1$.
Among $A_{i+k_A}$ and $B_{j+k_B}$ consider the one which begins earlier under substitution \solution:
this can be determined in $\Ocomp(1)$ by simply comparing the lengths, the length on the $A$-side of the equation is 
$k_A (|A_i| + |\sol X|)$ while $B$-side is $k_B (|B_j| + |\sol X|) + \ell$,
where $\ell$ is the remainder of \sol X that is compared with $A_i$.
Note that the test for the first letter of this word is not periodic,
so when we jump to it we skip the whole sequence of periodic tests.
We show that in $\Ocomp(1)$ time we can perform the tests for all letters before this word
and that the test right before the first test for $A_i$ is protected or misaligned.

Let $a = |A_i|$, $b = |B_j|$ and $x = |\sol{X}|$.
First consider the simpler case in which $a = b$.
Let $k = \min(k_A,k_B)$. Then the tests for $A_{i+1}$, \ldots, $A_{i + k - 1}$ are identical as for $A_i$,
and so it is enough to perform just the test for $A_i$ and $B_j$ and then jump right to $A_{i+k}$.

So let us now consider the case in which $a > b$.
Observe that when the whole $\sol{(B_jX)^\ell}$ is within $\sol{(A_iX)^3}$
then this can be tested in constant time in a naive way:
the length of $\sol{(A_iX)^3}$ is $3(a+x)$ while the length of $\sol{B_jX}^\ell$ is $\ell(b + x)$.
Hence $3(a+x) \geq \ell (b+x)$ and so $\ell \leq 3(a+x)/(b+x) \leq 3 \max(a/b,x/x) \leq 3 \bound$,
because $a/b$ is at most $\bound$.
Thus all tests for $\sol{(A_iX)^3}$ and $\sol{(B_jX)^\ell}$ can be done in $\Ocomp(\bound) = \Ocomp(1)$ time.

\begin{figure}
	\centering
		\includegraphics{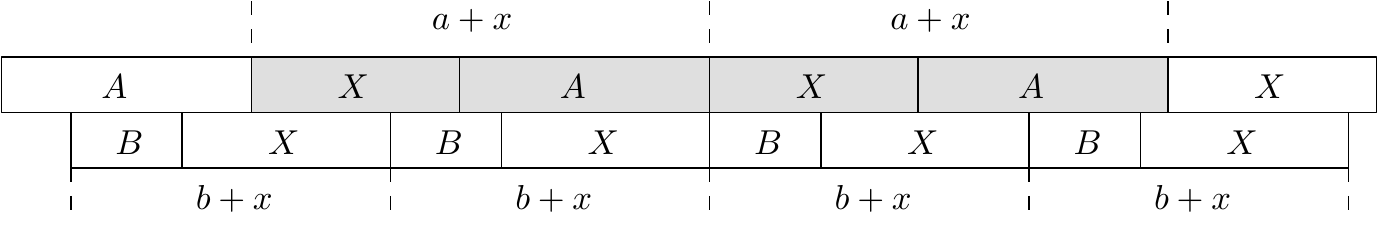}
	\caption{The case of $a > b$. The part of $\sol{(XA_i)^2}$ that has a period $a + x$ and $b + x$ is in grey.}
	\label{fig:periodical}
\end{figure}

So consider the remaining case, see Fig.~\ref{fig:periodical} for an illustration, when $k > 3$.
We claim that the tests for common part of $\sol{A_{i}X \cdots XA_{i+k-1}X}$ and $\sol{B_{j}X \cdots XB_{j+k-1}X}$
are successful if and only if
\begin{itemize}
	\item \sol{A_iX} and \sol{B_jX} have period $\gcd(a+x,b+x)$ \emph{and}
	\item the first $\gcd(a+x,b+x)$ tests for $\sol{A_{i}X \cdots XA_{i+k-1}X}$ and $\sol{B_{j}X \cdots XB_{j+k-1}X}$ are successful.
\end{itemize}
\textcircled{$\Rightarrow$}
First $\sol{XA_iXA_i}$ has period $x+a$.
However, it is covered with $\sol{(B_jX)^\ell}$, so it also has period $x+b$.
Since $x+a+x+b < 2x+2a$, it follows that also the $\gcd(x+a,x+b)$ is a period of $\sol{XA_iXA_i}$
and so also of \sol{A_iX} and thus also $\sol{B_jX}$.
The second item is obvious.

\noindent
\textcircled{$\Leftarrow$}
Since \sol{A_iX} and \sol{B_jX} have period $\gcd(a+x,b+x)$ also $\sol{A_{i}X \cdots XA_{i+k-1}X}$ and $\sol{B_{j}X \cdots XB_{j+k-1}X}$
have this period.
As the first $\gcd(a+x,b+x)$ tests for $\sol{A_{i}X \cdots XA_{i+k-1}X}$ and $\sol{B_{j}X \cdots XB_{j+k-1}X}$ are successful,
it follows that all the tests for their common part are.

So, to perform the test for the common part of $\sol{A_{i}X \cdots XA_{i+k_A-1}X}$ and $\sol{B_{j}X \cdots XB_{j+k_B-1}X}$
it is enough to:
calculate $p = \gcd(a+x,b+x)$,
test whether $\sol{A_iX}$, $\sol{B_jX}$ have period $p$ and then perform the first $p$ tests for
$\sol{A_{i}X \cdots XA_{i+k_A-1}X}$ and $\sol{B_{j}X \cdots XB_{j+k_B-1}X}$.
All of this can be done in $\Ocomp(1)$, since $p \leq a - b \leq \bound$
(note also that calculating $p$ can be done in $\Ocomp(1)$, as $\gcd(x+a,x+b)=\gcd(a-b,x+b)$ and $a-b \leq \bound$).

The case with $b > a$ is similar: in the special subcase we consider whether $\sol{(A_iX)^\ell}$ is within $\sol{X(B_jX)^3}$.
If so then the tests can be done in $\Ocomp(\bound)$ time.
If not, then we observe that the $\sol{XB_{j+1}XB_{j+2}}$ is covered by $\sol{(A_iX)^\ell}$.
So it the tests are successful, it has period both $x + b$ as well as $x + a$,
so it has period $\gcd(x+a,x+b)$. The rest of the argument is identical.

\begin{figure}
	\centering
		\includegraphics{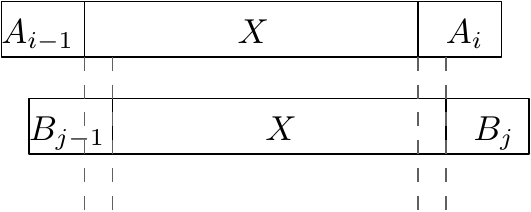}
	\caption{The test right before the first among the sequence of periodic tests. Since $A_i$ begins not later than $B_j$,
	$B_{j-1}$ ends not earlier than $A_{i-1}$.\label{fig:what_is_prev}}
\end{figure}

For the accounting, we would like to show that the test right before the first among the considered periodic tests is not periodic.
Observe, that as $A_i$ begins not later (under \solution) than $B_j$ it means that the last letter of $B_{j-1}$
is not earlier than the last letter of $A_{i-1}$, see Figure~\ref{fig:what_is_prev}.
So the previous test includes the last letter of $B_{j-1}$.
It is enough to show that this test is not failed, periodic, nor aligned.
\begin{description}
	\item[failed] If it is failed then also the test for the letters in $A_i$ are failed.
	\item[periodic] If it is periodic then this contradicts our choice that the test for the first letter of $A_i$
	is the first in the sequence periodic tests.
	\item[aligned] Since the first letter of $A_i$ is arranged against $XB_j$,
	in this case the last letter of $B_{j-1}$ needs to be arranged against the last letter of $A_{i-1}$.
	Then by the definition of the aligned test, $B_j = A_i$ and their first letters are at the same position.
	As by the assumption about the periodic tests we know that $A_{i+1} = A_i$ and $B_{j+1} = B_j$ we conclude that the test for the
	first letter of $A_i$ is in fact aligned, contradiction.
\end{description}
Hence, by Lemma~\ref{lem:no other test}, the test for the last letter of $B_{j-1}$ is either protected or misaligned.
Using the same accounting as in Lemma~\ref{lem:aligned is linear} we conclude that we spent at most $\Ocomp(n)$ time on all periodic tests.
\qedhere
\end{proof}

\subsubsection*{Proof of Lemma~\ref{lem: only letter tests}}
It is left to show that indeed we do not need to take into the account the time spent on comparing \sol X with \sol X on the other side of the equation.

\begin{proof}[proof of Lemma~\ref{lem: only letter tests}]
Recall that we only test solutions of the form $\sol X = a^k$.
Since we make the comparisons from left to the right in both \sol {\mathcal A_\ell} and \sol {\mathcal B_\ell}
then when we begin comparing letters from one \sol X with the other \sol X,
we in fact compare some suffix $a^\ell$ of $a^k$ with $a^k$.
Then we can skip those $a^\ell$ letters in $\Ocomp(1)$ time.
Consider the previous test, which needs to include at least one explicit letter.
Whatever type of test it was or whatever group of tests it was in, some operations were performed and this took $\Omega(1)$ time.
So we associate the cost of comparing \sol X with \sol X to the previous test, increasing the running time by at most a multiplicative constant.
\end{proof}

\subsection*{Open problems}
Is it possible to remove the usage of range minimum queries from the algorithm without increasing the running time?
Can the recompression approach be used to speed up the algorithms for the two variable word equations?
Can one use recompression approach also to better upper bound the number of solutions of an equation with a single variable?

\subsection*{Acknowledgements}
I would like to thank A.~Okhotin for his remarks about ingenuity of Plandowski's
result, which somehow stayed in my memory;
P.~Gawrychowski for initiating my interest in compressed membership problems
and compressed pattern matching, exploring which eventually led to this work
as well as for pointing to relevant literature~\cite{LohreySLP,MehlhornSU97};
J.~Karhum{\"a}ki, for his explicit question, whether the techniques of
local recompression can be applied to the word equations;
last not least, W.~Plandowski for his numerous comments and suggestions on the recompression applied to word equations.


\begin{thebibliography}{10}

\bibitem{rmq}
Omer Berkman and Uzi Vishkin.
\newblock Recursive star-tree parallel data structure.
\newblock {\em SIAM J. Comput.}, 22(2):221--242, 1993.

\bibitem{CharatonikPacholski}
Witold Charatonik and Leszek Pacholski.
\newblock Word equations with two variables.
\newblock In Habib Abdulrab and Jean-Pierre P{\'e}cuchet, editors, {\em
  IWWERT}, volume 677 of {\em LNCS}, pages 43--56. Springer, 1991.

\bibitem{twovarnew}
Robert D\k{a}browski and Wojciech Plandowski.
\newblock Solving two-variable word equations.
\newblock In Josep D\'{\i}az, Juhani Karhum{\"a}ki, Arto Lepist{\"o}, and
  Donald Sannella, editors, {\em ICALP}, volume 3142 of {\em LNCS}, pages
  408--419. Springer, 2004.

\bibitem{onevarold}
Robert D\k{a}browski and Wojciech Plandowski.
\newblock On word equations in one variable.
\newblock {\em Algorithmica}, 60(4):819--828, 2011.

\bibitem{FCPM}
Artur Je\.z.
\newblock Faster fully compressed pattern matching by recompression.
\newblock In Artur Czumaj, Kurt Mehlhorn, Andrew Pitts, and Roger Wattenhofer,
  editors, {\em ICALP (1)}, volume 7391 of {\em LNCS}, pages 533--544.
  Springer, 2012.

\bibitem{grammar}
Artur Je\.z.
\newblock Approximation of grammar-based compression via recompression.
\newblock In Johannes Fischer and Peter Sanders, editors, {\em CPM}, volume
  7922 of {\em LNCS}, pages 165--176. Springer, 2013.
\newblock full version at http://arxiv.org/abs/1301.5842.

\bibitem{fullyNFA}
Artur Je\.z.
\newblock The complexity of compressed membership problems for finite automata.
\newblock {\em Theory of Computing Systems}, 2013.
\newblock accepted for publication.

\bibitem{wordequations}
Artur Je\.z.
\newblock {Recompression: a simple and powerful technique for word equations}.
\newblock In Natacha Portier and Thomas Wilke, editors, {\em STACS}, volume~20
  of {\em LIPIcs}, pages 233--244, Dagstuhl, Germany, 2013. Schloss
  Dagstuhl--Leibniz-Zentrum fuer Informatik.

\bibitem{suffixarrays}
Juha K{\"a}rkk{\"a}inen, Peter Sanders, and Stefan Burkhardt.
\newblock Linear work suffix array construction.
\newblock {\em J. ACM}, 53(6):918--936, 2006.

\bibitem{lcpsuffixarrays}
Toru Kasai, Gunho Lee, Hiroki Arimura, Setsuo Arikawa, and Kunsoo Park.
\newblock Linear-time longest-common-prefix computation in suffix arrays and
  its applications.
\newblock In Amihood Amir and Gad~M. Landau, editors, {\em CPM}, volume 2089 of
  {\em LNCS}, pages 181--192. Springer, 2001.

\bibitem{onevarnew}
Markku Laine and Wojciech Plandowski.
\newblock Word equations with one unknown.
\newblock {\em Int. J. Found. Comput. Sci.}, 22(2):345--375, 2011.

\bibitem{LohreySLP}
Markus Lohrey and Christian Mathissen.
\newblock Compressed membership in automata with compressed labels.
\newblock In Alexander~S. Kulikov and Nikolay~K. Vereshchagin, editors, {\em
  CSR}, volume 6651 of {\em LNCS}, pages 275--288. Springer, 2011.

\bibitem{makanin}
G.~S. Makanin.
\newblock The problem of solvability of equations in a free semigroup.
\newblock {\em Matematicheskii Sbornik}, 2(103):147--236, 1977.
\newblock (in Russian).

\bibitem{MehlhornSU97}
Kurt Mehlhorn, R.~Sundar, and Christian Uhrig.
\newblock Maintaining dynamic sequences under equality tests in polylogarithmic
  time.
\newblock {\em Algorithmica}, 17(2):183--198, 1997.

\bibitem{onevarfirst}
S.~Eyono Obono, Pavel Goralcik, and M.~N. Maksimenko.
\newblock Efficient solving of the word equations in one variable.
\newblock In Igor Pr\'{\i}vara, Branislav Rovan, and Peter Ruzicka, editors,
  {\em MFCS}, volume 841 of {\em LNCS}, pages 336--341. Springer, 1994.

\bibitem{PlandowskiSTOC}
Wojciech Plandowski.
\newblock Satisfiability of word equations with constants is in {NEXPTIME}.
\newblock In {\em STOC}, pages 721--725. ACM, 1999.

\bibitem{PlandowskiFOCS}
Wojciech Plandowski.
\newblock Satisfiability of word equations with constants is in {PSPACE}.
\newblock {\em J.~ACM}, 51(3):483--496, 2004.

\bibitem{PlandowskiSTOC2}
Wojciech Plandowski.
\newblock An efficient algorithm for solving word equations.
\newblock In Jon~M. Kleinberg, editor, {\em STOC}, pages 467--476. ACM, 2006.

\bibitem{PlandowskiICALP}
Wojciech Plandowski and Wojciech Rytter.
\newblock Application of {Lempel}-{Ziv} encodings to the solution of word
  equations.
\newblock In Kim~Guldstrand Larsen, Sven Skyum, and Glynn Winskel, editors,
  {\em ICALP}, volume 1443 of {\em LNCS}, pages 731--742. Springer, 1998.

\bibitem{SLPaproxSakamoto}
Hiroshi Sakamoto.
\newblock A fully linear-time approximation algorithm for grammar-based
  compression.
\newblock {\em J. Discrete Algorithms}, 3(2-4):416--430, 2005.

\end{thebibliography}
\end{document}